\documentclass[letterpaper, 11pt]{article}

\usepackage{amsmath, amssymb, amsfonts, amsthm}
\usepackage{bm,bbm}
\usepackage{mathtools}
\usepackage{thm-restate}
\usepackage{mathrsfs}
\usepackage[left=0.75in,right=0.75in]{geometry}
\usepackage{setspace}
\usepackage[utf8x]{inputenc}
\usepackage{algpseudocode,algorithmicx,algorithm}
\usepackage[T1]{fontenc}
\usepackage{lmodern}
\usepackage{graphicx}
\usepackage{enumerate}
\usepackage[natural, dvipsnames, pdftex]{xcolor}
\usepackage{tikz}
\usepackage{verbatim}
\usepackage{hyperref}
\usepackage{cancel}
\usepackage[all]{xy}
\usepackage{mdframed}
\usepackage[capitalize]{cleveref}
\usepackage{todonotes}

\theoremstyle{plain}
\newtheorem{theorem}{Theorem}[section]
\newtheorem{lemma}[theorem]{Lemma}
\newtheorem{proposition}[theorem]{Proposition}
\newtheorem{corollary}[theorem]{Corollary}

\theoremstyle{definition}
\newtheorem{definition}[theorem]{Definition}

\theoremstyle{remark}
\newtheorem*{remark}{Remark}

\newcommand{\R}{\mathbb{R}}
\newcommand{\N}{\mathbb{N}}

\newcommand{\F}{\mathbb{F}}

\newcommand{\cC}{\mathcal{C}}
\newcommand{\cM}{\mathcal M}
\newcommand{\calX}{\mathcal{X}}
\newcommand{\calY}{\mathcal{Y}}
\newcommand{\cD}{\mathcal{D}}

\newcommand{\cB}{\mathcal{B}}
\newcommand{\cT}{\mathcal{T}}

\newcommand{\eps}{\varepsilon}

\newcommand{\supp}{\mathrm{supp}}

\newcommand{\spa}{\mathrm{span}}

\newcommand{\parens}[1]{\left( #1 \right)}
\newcommand{\brk}[1]{\left[ #1 \right]}
\newcommand{\curl}[1]{\left\{ #1 \right\}}

\newcommand{\Prop}{\mathop{\Pr}}
\newcommand{\E}{\mathbb{E}}
\newcommand{\Eop}{\mathop{\E}}

\newcommand{\wt}{\mathrm{wt}}
\newcommand{\one}{\vec{1}}
\newcommand{\perr}{p_{\mathrm{err}}}
\newcommand{\bv}{\bm{v}}
\newcommand{\bu}{\bm{u}}
\newcommand{\bw}{\bm{w}}
\newcommand{\bx}{\bm{x}}
\newcommand{\bz}{\bm{z}}
\newcommand{\ba}{\bm \alpha}

\newcommand{\bS}{\bm{S}}
\newcommand{\bc}{\bm{c}}
\newcommand{\Ber}{\mathrm{Ber}}
\newcommand{\Unif}{\mathrm{Unif}}
\newcommand{\rlc}{\bm \cC}

\begin{document}
\title{Threshold Rates of Code Ensembles: \\ Linear is Best}
\author{Nicolas Resch\thanks{Cryptology Group, Centrum Wiskunde \& Informatica. Research supported in part by ERC H2020 grant No.74079 (ALGSTRONGCRYPTO).} \\ \href{nicolas.resch@cwi.nl}{nicolas.resch@cwi.nl} \and Chen Yuan\thanks{School of Electronic Information and Electrical Engineering, Shanghai Jiao Tong University. Research supported in part by the National Natural Science Foundation of China under Grant 12101403, the National Natural Science Foundation of China under Grant 12031011 and
National Key Research and Development Project 2021YFE0109900.} \\ \href{chen_yuan@sjtu.cn.edu}{chen\_yuan@sjtu.cn.edu}}
\date{\today}
\maketitle

\begin{abstract}
In this work, we prove new results concerning the combinatorial properties of random linear codes. By applying the thresholds framework from Mosheiff et al. (FOCS 2020) we derive fine-grained results concerning the list-decodability and -recoverability of random linear codes.

Firstly, we prove a lower bound on the list-size required for random linear codes over $\F_q$ $\eps$-close to capacity to list-recover with error radius $\rho$ and input lists of size $\ell$. We show that the list-size $L$ must be at least $\frac{\log_q\binom{q}{\ell}-R}{\eps}$, where $R$ is the rate of the random linear code. This is analogous to a lower bound for list-decoding that was recently obtained by Guruswami et al. (IEEE TIT 2021B). As a comparison, we also pin down the list size of random codes which is $\frac{\log_q\binom{q}{\ell}}{\eps}$. This result almost closes the $O(\frac{q\log L}{L})$ gap left by Guruswami et al. (IEEE TIT 2021A). This leaves open the possibility (that we consider likely) that random linear codes perform better than random codes for list-recoverability, which is in contrast to a recent gap shown for the case of list-recovery from erasures (Guruswami et al., IEEE TIT 2021B).

Next, we consider list-decoding with constant list-sizes. Specifically, we obtain new lower bounds on the rate required for:
\begin{itemize}
	\item List-of-$3$ decodability of random linear codes over $\F_2$;
	\item List-of-$2$ decodability of random linear codes over $\F_q$ (for any $q$).
\end{itemize}
This expands upon Guruswami et al. (IEEE TIT 2021A) which only studied list-of-$2$ decodability of random linear codes over $\F_2$. Further, in both cases we are able to show that the rate is larger than that which is possible for uniformly random codes.

A conclusion that we draw from our work is that, for many combinatorial properties of interest, random linear codes actually perform \emph{better} than uniformly random codes, in contrast to the apparently standard intuition that uniformly random codes are best.
\end{abstract}

\section{Introduction} \label{sec:intro}
Coding theory is concerned with developing efficient means to make data robust to noise. The mathematical objects used for this purpose are \emph{(error-correcting) codes}, which are just subsets $\cC \subseteq \Sigma^n$, where $\Sigma$ is a finite alphabet of size $q$. It is often convenient to set $\Sigma=\F_q$, where $\F_q$ is the finite field of order $q$,\footnote{In this case, we will of course insist that $q$ be a prime power.} in which case we can insist that $\cC$ be a subspace of $\F_q^n$. We call such a code \emph{linear} and denote it $\cC \leq \F_q^n$. As we are mostly concerned with linear codes in the sequel we will always set $\Sigma=\F_q$.\footnote{For nonlinear codes this does potentially lose some generality, as the alphabet size in that case could be any integer. We do remark that our results concerning arbitrary codes hold for all integer $q$, but emphasizing this point is not relevant to our purposes.}

In order for a code to be useful for information transmission in noisy environments, we require $\cC$ to satisfy noise-resilience properties, which amounts to insisting that the codewords are ``difficult to confuse.'' A basic way to do this is to define a distance metric on $\F_q^n$ and then insist that the codewords are not too clustered. The standard choice is the (relative) \emph{Hamming distance} which is defined as $d(\vec x,\vec y) = \frac{1}{n}|\{i \in [n]:x_i\neq y_i\}|$ for $\vec x,\vec y \in \F_q^n$: in words, it is the fraction of coordinates on which the vectors $\vec x$ and $\vec y$ differ. The \emph{minimum distance} of a code is then the minimum distance between two distinct codewords, i.e., $\delta := \min\{d(\vec x,\vec y):\vec x,\vec y \in \cC, \vec x\neq \vec y\}$.

Beyond the minimum distance, there are other proxies for a code's noise-resilience that are widely studied. First and foremost, a popular relaxed notion of noise-resilience is provided by \emph{list-decodability}, which informally asks that the code not be ``too'' clustered around any one point. More precisely, a code is said to be $(\rho,L)$-list-decodable if there are never $L$ or more codewords that are all within distance $\rho$ of some vector $z \in \F_q^n$, i.e.,
\[
	\forall \vec z \in \F_q^n, ~ |\{\vec x \in \cC:d(\vec x,\vec z)\leq \rho \}| < L \ .
\]
The integer $L$ is called the \emph{list-size}. This notion, originally introduced by Elias and Wozencraft~\cite{Elias57,Wozencraft58}, finds uses within coding theory and beyond in, e.g., complexity theory~\cite{lipton1990efficient,babai1990bpp,sudan2001pseudorandom}, cryptography~\cite{kushilevitz1993learning}, and learning theory~\cite{goldreich1989hard}.

We will also investigate another relaxation of list-decoding: \emph{list-recovery}. Here, we are given a collection of input lists $S_1,\dots,S_n \subseteq \F_q$ of bounded size, and the requirement is that there are not too many codewords that agree too much with these input lists. More precisely, for an integer $\ell \leq q$ we require that
\[
	\forall \vec S = (S_1,\dots,S_n) \in \binom{\F_q}{\ell}^n, ~~ |\{\vec x \in \cC:d(\vec x,\vec S) \leq \rho\}| < L \ .
\]
In the above, we are denoting by $\binom{\F_q}{\ell}$ the family of all $\ell$-element subsets of $\F_q$, and we are extending the Hamming distance notation $d(\cdot,\cdot)$ via
\[
	d(\vec x,\vec S) = \frac{1}{n}|\{i \in [n]:x_i \notin S_i\}| \ .
\]
Note that $(\rho,1,L)$-list-recovery is equivalent to list-decoding, demonstrating that list-recoverability is indeed a generalization of list-decodability. While list-recovery was initially introduced as a stepping stone towards list-decoding~\cite{GuruswamiI01,GuruswamiI02,GuruswamiI03,GuruswamiI04} it has since found many new uses in theoretical computer science more broadly~\cite{guruswami2009unbalanced,INR10,NPR12,GNPRS13,HemenwayRW17,HemenwayW15}.

In order for a code to provide for efficient information transmission, we would like for the code's \emph{rate} to be as large as possible, which is a measure of the amount of information transmitted per symbol of a codeword. More precisely, the code's rate $R$ is defined as $\frac{\log_q|\cC|}{n}$; when the code is linear, this is simply $\frac{\dim(\cC)}{n}$.

At its core, coding theory is concerned with determining the achievable tradeoffs between a code's rate and its noise-resilience for various noise models. In this work, we focus upon the list-decodability and list-recoverability of codes. An important question we ask is how large the list-size $L$ must be for these tasks. This is useful in practice, as the main constraint on the run time of most list-decoding/recovery algorithms is due to the need to process the list. Further, in applications of list-recoverable codes to constructions of expanders~\cite{guruswami2009unbalanced} the quality of the expansion is directly governed by the list-size.

\paragraph{Random Ensembles of Codes.} As a stepping-stone towards a thorough understanding of the achievable tradeoffs (which is believed to be a very challenging problem), we take cues from much of the literature and study the behaviour of ``typical'' codes. That is, we sample codes of a prescribed rate according to natural distributions and investigate their list-decodability/recoverability. In particular, we consider random \emph{linear} codes, which are uniformaly sampled subspaces of $\F_q^n$ of the prescribed dimension. We also study uniformly random subsets of $\F_q^n$ of the prescribed size, which we call (uniformly) \emph{random codes}.

In our work, we endeavour to provide a more fine-grained understanding of the combinatorial properties of these code ensembles. In this way, we help to uncover the landscape of achievable parameters for various code properties of interest. Beyond its theoretical interest, many code constructions~\cite{GuruswamiI04,GuruswamiR08,HemenwayRW17,HemenwayW15} use (small) linear codes as a component, and better list-decodability/recoverability of these inner codes improves these constructions.

In our results, we highlight a (perhaps surprising) fact: for list-decoding/recovery, random linear codes seem to perform \emph{better} than uniformly random codes. On the one hand, even for the basic property of minimum distance it has already been observed that random linear codes (which achieve the Gilbert-Varshamov bound) outperform uniformly random codes. On the other hand, for problems such as list-decoding and list-recovery much of the literature appears to be focused on showing that random linear codes are ``not too much worse'' than uniformly random codes. We hope our work encourages the coding theory community to change perspective and endeavour to prove that random linear codes are in fact better.

\subsection{Our Results} \label{subsec:our-results}

\paragraph{List-Recoverability of Random Linear Codes.} As a first result, we provide a new lower bound on the list-size of random linear codes for list-recoverability. For context, we recall the list-recovery capacity theorem, which gives us some coarse-grained information regarding achievable tradeoffs. For an integer $1 \leq \ell < q$, error-radius $\rho \in (1-\ell/q)$ and $\eps>0$ it states the following:
\begin{itemize}
	\item If $R\leq 1-h_{q,\ell}(\rho)-\eps$, there exist $(\rho,\ell,L)$-list-recoverable codes with $L=O(\ell/\eps)$.
	\item If $R \geq 1-h_{q,\ell}(\rho)+\eps$, there \emph{do not} exist $(\rho,\ell,L)$-list-recoverable codes with $L = o(q^{\eps n})$.
\end{itemize}
In the above, the function $h_{q,\ell}(\cdot)$ is the $(q,\ell)$-entropy function; its precise definition is not important at the moment so we defer it to Section~\ref{sec:prelims}. Informally, when studying codes of rate $\eps$ below the capacity for a small $\eps>0$ we refer to them as \emph{capacity-approaching} and call $\eps$ as the \emph{gap-to-capacity}.

This already tells us that the capacity for $(\rho,\ell,L)$-list-recovery is $1-h_{q,\ell}(\rho)$ if we insist that $L$ be subexponential in $n$. However, we can ask for more fine-grained information: in particular, exactly how large must the list-size $L$ be as a function of $\eps$ and the other parameters?

%


For random linear codes, we prove the following lower bound.

\begin{theorem} [List-Recoverability Lower Bound for Random Linear Codes] \label{thm:list-rec-linear-informal}
	Let $1 \leq \ell \leq q$ be integers with $q$ a prime power and fix $\rho \in (0,1-\ell/q)$. Fix $\delta>0$. For sufficiently small $\eps>0$, a random linear code in $\F_q^n$ of rate $1-h_{q,\ell}(\rho)-\eps$ is whp \emph{not} $\left(\rho,\ell,\lfloor\frac{\log_q\binom{q}{\ell}-(1-h_{q,\ell}(\rho))}{\eps}-\delta\rfloor\right)$-list-recoverable.
\end{theorem}

For context, we consider the case of uniformly random codes. In this case, we obtain a tight result.

\begin{theorem}[List-Recoverability for Random Codes] \label{thm:list-rec-unif-informal}
	Let $1 \leq \ell \leq q$ be integers with $q$ a prime power and fix $\rho \in (0,1-\ell/q)$. Fix $\delta>0$. For sufficiently small $\eps>0$, a random code in $\F_q^n$ of rate $1-h_{q,\ell}(\rho)-\eps$ is whp \emph{not} $\left(\rho,\ell,\lfloor\frac{\log_q\binom{q}{\ell}}{\eps}-\delta\rfloor\right)$-list-recoverable.
	
	On the other hand, for any $\eps>0$ and $n$ sufficiently large, a random code in $\F_q^n$ of rate $1-h_{q,\ell}(\rho)-\eps$ is whp $\left(\rho,\ell,\lceil \frac{\log_q\binom{q}{\ell}}{\eps} \rceil +1\right)$-list-recoverable.
\end{theorem}
In this way, we pin down the list-recoverability for random codes to one of two or three possible values: $\lfloor\frac{\log_q\binom{q}{\ell}}{\eps}+0.99\rfloor,\lceil\frac{\log_q\binom{q}{\ell}}{\eps}\rceil$ (if it's different) or $\lceil \frac{\log_q\binom{q}{\ell}}{\eps} \rceil +1$.

Comparing Theorems~\ref{thm:list-rec-linear-informal}~and~\ref{thm:list-rec-unif-informal} we see that our lower bound on random linear codes is less than the precise bound we have on random codes. One could potentially draw the conclusion that Theorem~\ref{thm:list-rec-linear-informal} should be improved. However, we believe that it is in fact tight. For the case of list-decoding binary codes it has already been shown that random linear performs better than uniformly random, and the bounds we obtain are the natural generalizations of the (tight) results for that case. We therefore conjecture that Theorem~\ref{thm:list-rec-linear-informal} is indeed tight. This stands in stark contrast to \emph{erasure} list-recovery:\footnote{Here, the requirement is that for all subsets $S_1,\dots,S_n \subseteq \F_q$ where at least $(1-\rho)n$ of the $S_i$'s satisfy $|S_i|\leq \ell$ (and the others may be all of $\F_q$), the number of codewords in $S_1 \times \cdots \times S_n$ is less than $L$.} for this model, it is known that random linear codes can require lists of size $\ell^{\Omega(1/\eps)}$~\cite{GuruswamiLMRSW20} (at least, if the field has large characteristic), whereas the lists for random codes can be just $O(\ell/\eps)$. A summary of the state-of-the-art for list-recovery of RLCs and RCs is provided in Figure~\ref{fig:table-list-rec}.

\begin{remark}
	It might appear that our conjecture that random linear codes outperform random codes for list-recovery is contradicted by the result of~\cite{GuruswamiLMRSW20}. However, we emphasize that the capacity for erasure list-recovery is larger, so if a code is $\eps$-close to capacity for list-recovery from erasures for small $\eps>0$ it is above capacity for list-recovery from errors, the model we study. Hence, this lower bound does not contradict our conjecture. One can also consider the model where $\rho$ approaches the limit $1-\ell/q$ as is done in \cite{RudraW18}; in this case we still suspect that random linear codes outperform uniformly random codes, but this is just speculation and further investigation is required.
\end{remark}

\begin{figure}
	\centering
	\begin{tabular} {|c|c|c|c|c|}
		\hline
		Source & Model & Radius & Rate & List-size bound\\
		\hline
		\hline
		Folklore & RC & $\rho>0$ & $1-h_{q,\ell}(\rho)-\eps$ & $\leq O(\ell/\eps)$ \\
		\hline
		\cite{ZyablovP81} & RLC & $\rho>0$ & $1-h_{q,\ell}(\rho)-\eps$ & $\leq q^{O(\ell/\eps)}$ \\
		\hline
		\cite{RudraW18} & RLC & $\rho=1-\frac{\ell}{q}-\eta$ & $0.99(1-h_{q/\ell}(\alpha)-\log_q(\ell))$ & $\leq q^{O(\ln^2(\ell/\eta))}$ \\
		\hline
		\cite{GuruswamiLMRSW20} & RLC & $\rho = 0$ & $1-\log_q(\ell)-\eps$ & $\geq \ell^{\Omega(1/\eps)}$ \\
		\hline
		\hline
		Theorem~\ref{thm:list-rec-linear-informal} & RLC & $\rho > 0$ & $1-h_{q,\ell}(\rho)-\eps$ & $>  \frac{\log_q\binom{q}{\ell}-(1-h_{q,\ell}(\rho))}{\eps}$ \\
		\hline
		Theorem~\ref{thm:list-rec-unif-informal} & RC & $\rho>0$ & $1-h_{q,\ell}(\rho)-\eps$ & $\approx \frac{\log_q\binom{q}{\ell}}{\eps}$\\
		\hline
	\end{tabular}
	\caption{This table summarizes much of the work on the list-recoverability of random linear codes (RLC) and random codes (RC). The lower bound of \cite{GuruswamiLMRSW20} only applies when $q = p^{\Omega(1/\eps)}$ for a prime $p$, and in \cite{RudraW18} $\eta>0$ is viewed as a small constant. \cite{GuruswamiLMRSW20} also offers a similar lower bound for the case of list-recovery from erasures.}
	\label{fig:table-list-rec}
\end{figure}



\paragraph{List-decoding with small lists.} Next, we turn our attention to the challenge of list-decoding when the output list-size $L$ is a (small) constant. Thus, we are no longer in the regime where we can expect to approach the list-decoding capacity, and we are interested to know by how much we are required to back off if, say, $L=3,4$.

First, we consider the case where $L=4$ for the binary field, which we also refer to as \emph{list-of-$3$} decoding.
Here and throughout, we also use the following notation (which is slightly abusive): for $q \geq 2$ and nonnegative reals $x_1,\dots,x_t$ with $x_1+\dots+x_t \leq 1$, $H_q(x_1,\dots,x_t) = \sum_{i=1}^tx_i \log_q\frac{1}{x_i} + (1-x_1-\dots-x_t)\log_q\frac{1}{1-x_1-\dots-x_t}$.

We first prove the following possibility result for random linear codes. In the following,
\[
\cB_\rho = \curl{(x_1,x_2) \in \R^2:x_1+2x_2 \leq 4\rho, x_1+x_2 \leq 1, x_1,x_2 \geq 0} \ .
\]

\begin{theorem} [List-of-3 decoding Random Linear Binary Codes]\label{thm:list-of-3-rlc-informal}
	Let $\rho \in (0,5/16)$\footnote{If $\rho \geq 5/16$ it is known that there are no $(\rho,4)$-list-decodable codes with postive rate~\cite{alon2018list}.} and suppose
	\[
	R < 1-\max_{(x_1,x_2) \in \cB_\rho}\frac{H_2(x_1,x_2)+2x_1+x_2\log_2 3}{3} \ .
	\]
	Then a random linear code over $\F_q$ of rate $R$ is whp $(\rho,4)$-list-decodable.
\end{theorem}

For context, we also study the list-of-3 decodability of random codes over the binary alphabet. In this case, we can prove the following:

\begin{theorem} [List-of-3 decoding Random Binary Codes]\label{thm:list-of-3-random-informal}
	Let $\rho \in (0,5/16)$ and suppose
	\[
	R > 1-\max_{(x_1,x_2) \in \cB_\rho}\frac{1+H_2(x_1,x_2)+2x_1+x_2\log_2 3}{4}\ .
	\]
	Then a random code over $\{0,1\}$ of rate $R$ is whp \emph{not} $(\rho,4)$-list-decodable.
	
	On the other hand, if
	\[
	R < 1-\max_{(x_1,x_2) \in \cB_\rho}\frac{1+H_2(x_1,x_2)+2x_1+x_2\log_2 3}{4}\ ,
	\]
	then a random code over $\{0,1\}$ is whp $(\rho,4)$-list-decodable.
\end{theorem}

As $\frac{1+F}{4} \geq \frac{F}{3}$ whenever $F \leq 3$, we see that the bound in Theorem~\ref{thm:list-of-3-rlc-informal} is greater than the bound from Theorem~\ref{thm:list-of-3-random-informal}. Using terminology that we later make precise, we see that the \emph{threshold rate} for list-of-3 decoding binary random linear codes strictly exceeds that of binary random codes.

Next, we study list-of-$2$ decoding over alphabets of size $q>2$. And again, our theorems demonstrate that random linear codes strictly outperform random codes. Define
\[
\cD_{\rho}:=\{(x_1,x_2) \in \R^2:x_1+x_2 \leq 3\rho, x_1+x_2 \leq 1,x_1,x_2 \geq 0\} .
\]
\begin{theorem} [List-of-2 decoding Random Linear $q$-ary Codes]\label{thm:list-of-2-rlc-informal}
	Let $\rho \in (0,1/3)$ and suppose
	\[
	R < 1-\max_{(x_1,x_2) \in \cD_{\rho}}\frac{H_q(x_1,x_2)+x_1\log_q 3(q-1)+x_2\log_q (q-1)(q-2)}{2} \ .
	\]
	Then a random linear code over $\F_q$ of rate $R$ is whp $(\rho,3)$-list-decodable.
\end{theorem}

\begin{theorem} [List-of-2 decoding Random $q$-ary Codes]\label{thm:list-of-2-random-informal}
	Let $\F_q$ be an alphabet of size $q$. Let $\rho \in (0,1/3)$ and suppose
	\[
	R > 1-\max_{(x_1,x_2) \in \cD_{\rho}}\frac{1+H_q(x_1,x_2)+x_1\log_q 3(q-1)+x_2\log_q (q-1)(q-2)}{3} \ .
	\]
	Then a random code over $\F_q$ of rate $R$ is whp \emph{not} $(\rho,3)$-list-decodable.
	
	On the other hand, if
	\[
	R < 1-\max_{(x_1,x_2) \in \cD_{\rho}}\frac{1+H_q(x_1,x_2)+x_1\log_q 3(q-1)+x_2\log_q (q-1)(q-2)}{3} \ ,
	\]
	then a random code over $\F_q$ is whp $(\rho,3)$-list-decodable.
\end{theorem}


Again, we can see that the bound from Theorem~\ref{thm:list-of-2-rlc-informal} is greater than the bound from Theorem~\ref{thm:list-of-2-random-informal}. We therefore conjecture that this phenomenon of random linear codes outperforming random codes extends to more values of $L$. To provide more evidence for this conjecture, we extend an argument for binary random linear codes of \cite{GuruswamiHSZ02,LiW18} to larger values of $L$, and by comparing it to a computation of the threshold rate for random binary codes, show that for many parameter regimes of interest we do indeed have random linear codes outperforming random codes.

\subsection{Techniques} \label{subsec:techniques}

In order to obtain our results, we rely on a recently developed toolkit for proving threshold rates for combinatorial properties of random (linear) codes. The toolkit for random linear codes was developed by Mosheiff et al.~\cite{mosheiff2020ldpc} on the way to proving that LDPC codes achieve list-decoding capacity; recent works~\cite{GuruswamiLMRSW20,guruswami2021punctured} have found further uses for the techniques in investing combinatorial properties of random linear codes. An analogous threshold toolkit for \emph{random} codes was provided in~\cite{guruswami2021threshold}.

Broadly speaking, the techniques of \cite{mosheiff2020ldpc,guruswami2021threshold} apply when considering a property of codes defined by forbidding a family of ``bad'' subsets, each of which have constant cardinality (independent of $n$). For example, the property of $(\rho,L)$-list-decodability is defined by forbidding all $L$-element subsets $B=\{x_1,\dots,x_L\}$ of a Hamming ball $B(z,\rho) = \{x \in \F_q^n:d(x,z) \leq \rho\}$ from appearing in the code. In \cite{mosheiff2020ldpc}, it is proved that for any such local property there is a \emph{threshold} rate $R^*$ such that:
\begin{itemize}
	\item If $R<R^*$, a random linear code satisfies the property with high probability;
	\item If $R>R^*$, a random linear code fails to satisfy the property with high probability.
\end{itemize}
The theorem furthermore characterizes the threshold rate $R^*$ as the solution to a certain optimization problem. In this work, we endeavour to compute new bounds on the threshold rate $R^*$ for various properties of interest.

In the remainder, we provide intuition for the characterization of the threshold rate from \cite{mosheiff2020ldpc}. First, we identify subsets $B \subseteq \F_q^n$ of size $L$ with the matrix in $\F_q^{n \times L}$ whose columns are given by $B$ (the choice of ordering is immaterial), and we say that a matrix $M$ is contained in a code $\cC$ if $\cC$ contains all of $M$'s columns. For a collection of matrices $\cM \subseteq \F_q^{n \times L}$, we would like to compute the threshold rate $R^*$ for ``$\cM$-freeness,'' i.e., the code property of not containing a matrix in $\cM$.


As we are interested in list-decoding/recovery, we define a set of matrices $\cM$ such that if $\cC$ contains a matrix from $\cM$ then $\cC$ is not list-decodable/recoverable. We say that the collection $\cM$ is ``bad'' for list-decoding/recovery. As intuition, for list-decoding we can just take the set of matrices where each column lies in some ball $B(z,\rho)$. Next, we would like to show that $\cM$ is ``abundant'' in the sense that it is very likely that $\cC$ contains a matrix $M \in \cM$. In other words, if $X_M$ denotes the indicator random variable for the event $M \subseteq \cC$, then we should expect $X_{\cM}:=\sum_{M \in \cM}X_M \geq 1$.



It is relatively easy to compute $\E[X_{\cM}]$ and see when it exceeds $1$; however, to conclude that $X_{\cM}$ is likely to be large one needs a concentration bound. Such a bound is often provided by estimating the variance of $X_{\cM}$. Broadly construed, \cite{mosheiff2020ldpc} applies the second moment method to demonstrate that there is really only one reason that $X_{\cM}$ would fail to be concentrated: it is because for some compressing matrix $A \in \F_q^{L \times L'}$ with $L' \leq L$ the set $\{MA:M \in \cM\}$ is too small.


\paragraph{List-Recovery.} First, we endeavour to prove a lower bound on the list-size for list-recovery. This means that we need to say that if the list-size is too small then the random linear code quite likely contains a matrix from a set $\cM$ of bad matrices for list-recovery. In light of the above, to conclude our argument we need to show that for any compressing matrix $A$, the set $\{MA:M \in \cM\}$ remains large.

To do this, we use information-theoretic techniques: we identify each of our bad matrices $M \in \cM$ with an appropriate \emph{type}, which is a distribution $\tau \sim \F_q^L$ defined as the empirical distribution of $M$'s rows. A lower bound on $\{MA:M \in \cM\}$ is then implied by a lower bound on the entropy of the random variable $A\vec{\bu}$ for $\vec{\bu} \sim \tau$. We are also free to choose the type $\tau$ which is ``bad'' for a certain property, in the sense that if a code contains a matrix of type $\tau$ then it fails to satisfy the property.

For the case of $(\rho,\ell,L)$-list-recovery, the following type is bad: one samples uniformly $\bS \in \binom{\F_q}{\ell}$ and then outputs $\vec \bu = (\bu_1,\dots,\bu_L) \in \F_q^L$, where each $\bu_i$ is independently uniform over $\bS$ with probability $1-\rho$ and uniform over $\F_q\setminus\bS$ otherwise. It thus follows that a lower bound on $\{AM:M \in \cM\}$ is implied by a lower bound on the entropy of the random variable $A\vec \bu$ for $\vec{\bu} \sim \tau$.

Obtaining this lower bound requires a rather lengthy argument; we overview the main ideas now. We begin by partitioning the coordinates of $A\vec \bu$ into subsets $J_1,\dots,J_k \subseteq [L']$, where each $J_i$ depends on at least 2 ``fresh'' coordinates from $\vec \bu$, along with (perhaps) a set of leftover coordinates $J_{k+1}$. We then provide two arguments depending on the maximum size of a part. If, say, $|J_1|$ is large, then we can show that $(A\vec \bu)_{J_1}$ already experiences a large entropy increase. This is shown by demonstrating that these coordinates alone already allow us to nontrivially guess the subset $\bS$. Otherwise, we argue that all the parts provide a nontrivial increase in the entropy, and since there must be a large number of parts in this case, by summing over all the parts we provide an adequate lower bound.

This result generalizes the list-decoding lower bound that was provided in \cite[Theorem IV.1]{GuruswamiLMRSW20}. The argument in that paper exploited the fact that a sample from the bad type for list-decoding has a simpler structure: it looks like $\vec{\bv} + \ba\vec{1}$, where $\vec{\bv}$ is a $q$-ary Bernoulli random variable and $\ba \in \F_q$ is uniformly random. In our case, we do not have this nice linear structure,\footnote{One might be tempted to look at $\vec \bv + \vec \bw$ where $\vec \bv$ is $q$-ary Bernoulli and $\vec \bw$ is uniform over $\bS$, but note that for $\ell-1$ choices for $\bv_i \in \F_q^*$ the sum $\bv_i+\bw_i$ still lies in $\bS$.} making the analysis more intricate.

\paragraph{List-Decoding with Small Lists.} For our results concerning list-decoding with small lists, we again use the thresholds framework. In this case, we need to consider \emph{any} type that is bad for $(\rho,3)$ or $(\rho,4)$-list-decoding. For these small values of $L$, we are able to identify the linear map $A$ which leads to the maximum relative entropy $\frac{H_q(A\tau)}{\dim(A\tau)}$: in each case, it is given by the map sending $(x_1,\dots,x_L) \mapsto (x_1-x_L,\dots,x_{L-1}-x_L)$.

To provide the proof, we break up the vector spaces based on the number of distinct coordinates of the entries, and observe that a type which is bad for list-decodability can only put so much probability mass on each of these parts. To conclude, we rely on the concavity of the entropy function as well as some combinatorial reasoning concerning the subspaces of $\F_2^4$ and $\F_q^3$. Even for these small values of $L$ we need to be quite careful to avoid a massive explosion in the number of cases to consider, as we must look at all compressing linear maps $A$.


\paragraph{Random Codes.} For the case of random codes, we can compute the threshold rates for all the properties of interest in a relatively straightforward way, as the characterization from \cite{guruswami2021threshold} does not require us to consider any sort of compressing mapping on the types. Quite notably, in all cases we see that random linear codes appear to perform better than random codes. This is perhaps in contrast to commonly held beliefs: a main goal of our work is to disseminate this counterintuitive phenomenon.

\subsection{Related Work} \label{sec:related-work}

In \cref{subsec:techniques} we outlined the works \cite{mosheiff2020ldpc,GuruswamiLMRSW20,guruswami2021threshold} which developed and studied the thresholds toolkit that we apply. In this section, we provide more context for the study of random linear codes and their list-decodability/recoverability. In what follows, $q$ always denotes the alphabet size and $\eps$ the ``gap-to-capacity'' for a capacity-approaching code.

\paragraph{List Size Lower Bounds for Random (Linear) Codes.} As we provide lower bounds for list-recovery of random linear codes, we briefly survey the known lower bounds for list-decoding. First, Guruswami and Narayanan~\cite{guruswami2010lower} showed that capacity-approaching random (linear) codes require lists of size $\Omega_{\rho,q}(1/\eps)$: by inspecting the proof one can note that the implied constant tends to $0$ as $\rho \to 1-1/q$, or if $q \to \infty$. While on the surface their approach appears very different to ours, their use of a second-moment method is akin to the proofs underlying the thresholds framework from \cite{mosheiff2020ldpc}, so the approaches are in fact somewhat similar. Later, Li and Wootters \cite{LiW18} gave a $\sim 1/\eps$ list-size lower bound for capacity-approaching random codes. Again, the argument relies on the second-moment method.

In \cite{GuruswamiLMRSW20}, a lower bound for the list-decodability of capacity-approaching random linear codes is given, showing that lists of size $\sim\frac{h_q(\rho)}{\eps}$ are required: our list-recovery list-size lower bound is a generalization of this result. Lastly, in \cite{guruswami2021threshold} the threshold rate for $(\rho,2)$-list-decodability is computed, providing a lower bound and an upper bound: this segues us nicely into a discussion of the work on computing upper bounds on list-sizes.

\paragraph{List Size Upper Bounds for Random Linear Codes.} There has been a long line of work~\cite{ZyablovP81,GuruswamiHSZ02,GHK11,CheraghchiGV13,Wootters13,RudraW14a,RudraW18,LiW18,guruswami2021threshold} studying the list-decodability of capacity-approaching random linear codes, and we now highlight some relevant results. First, Zyablov and Pinkser~\cite{ZyablovP81} demonstrated that capacity-approaching RLCs are indeed $(\rho,L)$-list-decodable, albeit with $L = q^{\Omega(1/\eps)}$. Subsequent work has endeavoured to prove list-decodability with $L = O(1/\eps)$. The existence of such linear codes over $\F_2$ was first demonstrated by \cite{GuruswamiHSZ02}; later, \cite{LiW18} showed that this holds with high probability for randomly sampled linear codes, and subsequently \cite{GuruswamiLMRSW20} showed this is true for \emph{average-radius}\footnote{In this model, it is required that the code does not contain $L$ points whose average distance from a centre is less than $\rho$. Thus, it is a stricter requirement than standard list-decoding.} list-decoding.

As for larger alphabets, \cite{GHK11} showed that lists of size $O_{\rho,q}(1/\eps)$ do indeed suffice for random linear codes. We further remark that their argument uses a certain Ramsey-theoretic concept called a 2-increasing sequence to choose the order in which to reveal coordinates, which is vaguely reminiscent of the ``fresh'' coordinates that we have defined by the $J_i$'s in our list-recovery lower bound argument. A drawback of this work is that the implied constant in the $O_{\rho,q}(\cdot)$ notation degrades as $\rho \to 1-1/q$ or if $q$ grows too large. In light of this restriction, a line of works~\cite{CheraghchiGV13,Wootters13,RudraW14a} has studied the ``high noise regime,'' where $\rho = 1-1/q-\eta$ and one endeavours to show that lists of size $O(1/\eta^2)$ suffice for codes of rate $\Omega(\eta^2)$. These results are still not quite optimal in the sense that the implied constants (even for the rate) lag behind the parameters achievable by random codes. Lastly, for list-recoverability with input list-size $\ell$ it appears that the best upper bound on the list-size is due to \cite{RudraW18}, where it is shown that lists of size $(q\ell)^{O(\log(\ell)/\eps)}$ suffice.

\paragraph{Lower Bounds for List Sizes of Arbitrary Codes.} While we exclusively study random (linear) codes, we view these as a proxy for determining the actual achievable tradeoffs. As lists of size $\Theta(1/\eps)$ are required for random codes, it is natural to wonder if all capacity-approaching $(\rho,L)$-list-decodable codes require lists of size $\Omega(1/\eps)$. Blinovsky~\cite{Blinovsky86,blinovsky2005code} has shown a lower bound of $\Omega_\rho(\log(1/\eps))$. In the high noise regime, viz., $\rho=1-1/q-\eta$, Guruswami and Vadhan~\cite{guruswami2010lower} provided a $\Omega_q(1/\eta^2)$ lower bound on the list size. Lastly, for \emph{average-radius} list-decoding Guruswami and Narayanan~\cite{guruswami2014combinatorial} proved a $\Omega_\rho(1/\sqrt{\eps})$ lower bound.

\subsection{Open Problems} \label{subsec:open-probs}

In this work, we have progressed our understanding of combinatorial properties of random (linear) codes. A main conclusion of our work is that for list-decoding/recovery, random linear codes perform better.\footnote{For list-recovery, we admittedly only provide some evidence in this direction.}

There are many open problems which remain to be studied and we list some below.
\begin{itemize}
	\item Provide the corresponding upper bounds on the threshold rate for $(\rho,4)$-list-decoding binary random linear codes, and the threshold rate for $(\rho,3)$-list-decoding $q$-ary random linear codes.
	\item Provide the corresponding lower bound on the threshold rate for $(\rho,\ell,L)$-list-recovery in the capacity-approaching regime. In fact, for $q > 2$, the threshold rate for $(\rho,L)$-list-decoding is still open. This is quite likely a very challenging problem; the only tight argument we have is due to \cite{GuruswamiHSZ02,LiW18} (see also \cite{GuruswamiLMRSW20}) which only applies to list-decoding over the binary field, and this argument appears too ``rigid'' to apply in more generality.
	\item Get a better understanding for \emph{worst-case} codes. In particular, to the best of our knowledge the \emph{Plotkin points} for $(\rho,L)$-list-decoding for $q > 2$ are not known. That is, compute the minimum value $\rho^*$ such that for all $\rho > \rho^*$, there are no $q$-ary $(\rho,L)$-list-decodable code families with positive rate. (Recent work~\cite{zhang2020generalized} expresses the Plotkin point as a solution to a certain optimization problem, but we do not see how to extract a simple expression from this.)
\end{itemize}

\subsection{Organization} \label{subsec:organization}

In the subequent section, we introduce the necessary notations and definitions that we will use in this work, along with the tools from \cite{mosheiff2020ldpc,guruswami2021threshold} that we apply. In \cref{sec:list-rec}, we provide our lower bound on the list-size for the list-recoverability of random linear codes which approach capacity. In \cref{sec:list-dec}, we lower bound the threshold rate for list-of-2 decoding (for general $q$) and list-of-3 decoding (in the binary case). We also compare random linear codes to random codes over the binary alphabet for more values of $L$.

\section{Prelimaries} \label{sec:prelims}
\paragraph{Miscellaneous Notations.} 
For an integer $n \geq 1$, we denote $[n] := \{1,2,\dots,n\}$. For a set $X$ we denote by $\binom{X}{\ell}$ the family of all subsets of $X$ with $\ell$ elements, and similarly $\binom{X}{\leq \ell}$ denotes the family of all subsets of $X$ with $\leq \ell$ elements. Throughout, $\F_q$ denotes the finite field with $q$ elements, for $q$ a prime power.

For clarity, vectors are typically denoted with an arrow overtop. Given a vector $\vec x \in \F_q^n$ and a subset $I \subseteq [n]$ we denote by $\vec x_I$ the length $|I|$ vector $(x_i:i \in I) \in \F_q^{|I|}$. We reserve $\one$ for the all-$1$'s vector; if we wish to emphasize its length we subscript it, i.e., $\one_D$ is the all-$1$'s vector of length $D$. Random variables are typically written in boldface, e.g., $\bm x, \bm y$, etc. In particular, random vectors are denoted, e.g., $\vec \bu$. 

\paragraph{Coding Theory Terminology.} A \emph{code} $\cC$ is a subset of $\F_q^n$ for $\F_q$ the finite field of order $q$, a prime power. Elements $\vec c \in \cC$ are called \emph{codewords}, the integer $n$ is the \emph{block-length}, and the integer $q$ is the \emph{alphabet size}; such a code is also called \emph{$q$-ary}. When $q=2$ the code is deemed \emph{binary}. We are typically interested in \emph{linear} codes, which are $\cC \leq \F_q^n$, i.e., they are subspaces. The \emph{rate} of a code $\cC$ is $R = R(\cC) := \frac{\log_q|\cC|}{n}$ and its minimum distance is $\delta = \delta(\cC) := \min\{d(\vec x,\vec y):\vec x \neq \vec y,\vec x,\vec y \in \cC\}$, where $d(\vec x,\vec y) = \frac{1}{n}|\{i \in [n]:x_i \neq y_i\}|$ is the (relative) Hamming distance from $\vec x$ to $\vec y$. We also slightly extend this notation as follows: for a vector $\vec x \in \F_q^n$ and a tuple of subsets $\vec S = (S_1,\dots,S_n)$, $S_i \subseteq \F_q$, we define $d(\vec x,\vec S) := \frac{1}{n}|\{i \in [n]:x_i \notin S_i\}|$, i.e., the fraction of coordinates $i$ for which $\vec x$ ``disagrees'' with the corresponding subset of $\vec S$.

A \emph{random linear code} of rate $R$ is a uniformly random subspace of $\F_q^n$ of dimension $Rn$.\footnote{In fact, there are different ways to sample linear codes. For concreteness, we typically implicitly use the model where a random parity check matrix $\bm{H} \in \F_q^{(1-R)n \times n}$ is sampled and we output $\rlc=\ker(\bm{H})$. Of course, there is a small chance $\rlc$ has rate larger than $R$, but as this probability is exponentially small in $n$ it is immaterial to our conclusions. We also briefly use the model where a random $\bm{G} \in \F_q^{Rn \times n}$ is sampled and we output $\rlc = \mathrm{im}(\bm{G})$.} As this concept will arise regularly in this work, we occasionally use the abbreviation \emph{RLC}. A \emph{random code} of rate $R$ is a random subset of $\F_q^n$ obtained by including each element independently with probability $q^{(R-1)n}$.\footnote{By Chernoff bounds, such a code as rate $R\pm o(1)$ with high probability.} For this concept, we use the abbreviation \emph{RC}. 


\subsection{List-decodability and List-recoverability} \label{subsec:list-defns}

In this work, we study combinatorial properties of linear codes. Of primary interest to us are list-decodability and list-recoverability, which we now define. 

\begin{definition} [List-decodability] \label{defn:list-dec}
	Let $\rho \in (0,1-1/q)$ and $L \geq 1$. A code $\cC \subseteq \F_q^n$ is called $(\rho,L)$\emph{-list-decodable} if for all $\vec z \in \F_q^n$,
	\[
	|\{\vec c \in \cC:d(\vec c,\vec z) \leq \rho\}| < L \ .
	\]
\end{definition}
We also use the terminology ``list-of-$L$-decoding'' for $(\rho,L+1)$-list-decoding, e.g., list-of-2-decoding corresponds to $(\rho,3)$-list-decoding.

The list-decoding \emph{capacity} is the value $R^*(\rho)$ such that for any $R<R^*(\rho)$ there exists $L>1$ such that infinite families of $(\rho,L)$-list-decodable codes of rate at least $R$ exist, but for any $R>R^*(\rho)$ such an infinite family does not exist. It is known that
\[
R^*(\rho) = 1-h_q(\rho) \ ,
\]
where
\[
h_q(\rho) = \rho\log_q\frac{q-1}{\rho} + \log_q\frac{1}{1-\rho}
\]
is the $q$-ary entropy function. 

\begin{definition} [List-recoverability] \label{defn:list-rec}
	Let  $\rho \in (0,1-1/q)$, $1 \leq \ell \leq q$ and $L \geq 1$. A code $\cC \subseteq \F_q^n$ is called $(\rho,\ell,L)$\emph{-list-recoverable} if for all tuples of subsets $\vec S = (S_1,\dots,S_n) \in \binom{\F_q}{\leq \ell}^n$,
	\[
	|\{\vec c \in \cC:d(\vec c,\vec S) \leq \rho\}| < L \ .
	\]
\end{definition}

In analogy to the list-decoding capacity, the \emph{list-recovery capacity} is the value $R^*(\rho,\ell)$ such that for any $R<R^*(\rho,\ell)$ there exists $L>1$ such that infinite families of $(\rho,\ell,L)$-list-recoverable codes of rate at least $R$ exist, but for any $R>R^*(\rho,\ell)$ such an infinite family does not exist. It is known that
\[
R^*(\rho,\ell) = 1-h_{q,\ell}(\rho) \ ,
\]
where
\[
h_{q,\ell}(\rho) = \rho\log_q\frac{q-\ell}{\rho} + (1-\rho)\log_q\frac{\ell}{1-\rho}
\]
is the $(q,\ell)$-entropy function.

\subsection{Information-Theoretic Concepts} \label{sec:info-theoretic-concepts}

For a random variable $\bm x$ over a domain $\calX$ we denote its entropy by
\[
H(\bm x) = \sum_{x \in \calX}\Pr[\bm x=x]\log\frac{1}{\Pr[\bm x=x]} \ ,
\]
where we use the convention $0\log\frac{1}{0} = 0$. If $\tau$ is a distribution then we define $H(\tau)$ to be the entropy of a random variable distributed according to $\tau$. 

Given another random variable $\bm y$ supported on a set $\calY$, the \emph{conditional entropy} of $\bm x$ given $\bm y$ is
\[
H(\bm x|\bm y) = \Eop_{y \sim \bm y}[H(\bm x|\bm y=y)] = \sum_{x \in \calX,y \in \calY} \Pr[\bm x=x,\bm y =y]\log\frac{\Pr[\bm x =x]}{\Pr[\bm x =x,\bm y =y]} \ .
\]
Intuitively, this is the expected amount of entropy remaining in $\bm x$ after revealing $\bm y$. Conditional entropy satisfies the \emph{chain rule} $H(\bm x,\bm y) = H(\bm x|\bm y) + H(\bm y)$, which can be extended by induction to larger collections of random variables.

We also use the notion of \emph{mutual information}, which is a measure of the amount of information one random variable gives about another and is defined as follows:
\[
I(\bm x;\bm y) = H(\bm x)-H(\bm x|\bm y) = H(\bm y)-H(\bm y|\bm x) = H(\bm x,\bm y) - H(\bm x) - H(\bm y) \ .
\]
(The equalities are justified by the chain rule.) We also consider the \emph{conditional mutual information}, defined as follows:
\[
I(\bm x;\bm y|\bm z) = H(\bm x|\bm z)-H(\bm x|\bm y,\bm y) = H(\bm y|\bm z)-H(\bm y|\bm x,\bm z) = H(\bm x,\bm y|\bm z)-H(\bm x|\bm z)-H(\bm y|\bm z) \ ,
\]
where $\bm z$ is another random variable.

Conditional entropy, mutual information and conditional mutual information all satisfy the \emph{data-processing inequality}: for any function $f$ supported on $\calY$ (the support of $\bm y$), we have
\begin{align*}
	&H(\bm x|f(\bm y)) \geq H(\bm x|\bm y) \ , I(\bm x;\bm y) \geq I(\bm x;f(\bm y)) \ , I(\bm x;\bm y|\bm z) \geq I(\bm x;f(\bm y)|\bm z) \ .
\end{align*}
We will also use \emph{Fano's inequality}, which makes precise the intuition that if $\bm y$ allows one to guess the value of $\bm x$ with good probability, then the conditional entropy $H(\bm x,\bm y)$ cannot be too large. 

\begin{theorem} [Fano's Inequality.] \label{thm:fano}
	Let $\bm x$ be a random variable supported on $\calX$, $\bm y$ a random variable supported on $\calY$ and $f:\calY \to \calX$. Define $\perr := \Pr[f(\bm y) \neq \bm x]$. Then,
	\[
	H(\bm x|\bm y) \leq h(\perr) + \perr \cdot \log(|\calX|-1) \ .
	\]
\end{theorem}
When we wish to change the base of the logarithm with which the entropy or mutual information is computer, the desired base is subscripted. That is,
\[
	H_q(\bm x) := \frac{H(\bm x)}{\log q} \ , ~~~~~~~~~~ I_q(\bm x;\bm y) := \frac{I(\bm x;\bm y)}{\log q} \ ,
\]
and similarly for the conditional versions of these quantities. Finally, as a slight abuse of notation, we also write 
$$
H_q(x_1,\ldots,x_t)=\sum_{i=1}^t x_i\log_q \frac{1}{x_i} + (1-x_1-\dots-x_t)\log_q\frac{1}{1-x_1-\dots-x_t}
$$
if $x_1,\ldots,x_t$ are positive numbers satisfying $\sum_{i=1}^t x_i\leq 1$. (We caution that for $q > 2$, $H_q(x)\neq h_q(x)$.)

\subsection{Thresholds} \label{subsec:threshold-defns}

We now introduce the specialized notations and tools that we will need in order to apply the machinery of \cite{mosheiff2020ldpc}. First, for a distribution $\tau \sim \F_q^b$ and a linear map $A:\F_q^b \to \F_q^c$, we let $A\tau$ denote the distribution of the random vector $A\vec \bu$ for $\vec\bu \sim \tau$. In more detail, $A\tau$ has the following probability mass function:
\[
\Prop_{\vec \bv \sim A\tau}[\vec \bv=\vec y] = \sum_{\vec x \in A^{-1}(\vec y)}\Prop_{\vec \bu \sim \tau}[\vec \bu=\vec x] \ .
\]
While we are generally concerned with understanding the probability that certain ``bad sets'' lie in our code, it is in fact more convenient to work with matrices. For a matrix $M \in \F_q^{n \times b}$ and a code $\cC \subseteq \F_q^n$ we say that $\cC$ contains $M$ if
the columns of $M$ are contained in $\cC$.

Every matrix is assigned a \emph{type}, and the type of a matrix is determined by the matrix's empirical row distribution as follows:

\begin{definition} [$\tau_M,\dim(\tau),\cM_{n,\tau}$]
	For a matrix $M \in \F_q^{n \times b}$, we define its \emph{type} $\tau_M$ to be the distribution given by the empirical distribution of $M$'s rows. That is, for all $\vec v \in \F_q^b$ we have
	\[
	\tau_M(\vec v) := \frac{|\{i \in [n]:i\text{th row of } M \text { equals }\vec v\}|}{n} \ .
	\]
	For a distribution $\tau$ on $\F_q^b$, $\dim(\tau)$ denotes the dimension of the span of $\tau$'s support, i.e.,
	\[
	\dim(\tau) := \dim(\mathrm{span}(\supp(\tau))) .
	\]
	We denote by $\cM_{n,\tau}$ the set of all matrices in $\F_q^{b \times n}$ with empirical row distribution $\tau$. We call a type $\tau$ \emph{$b$-local} if $\tau \sim \F_q^b$; note that a $b$-local type has $\dim(\tau) \leq b$.
\end{definition}

\begin{remark} \label{rmk:technical-o(1)}
	Technically, for a distribution $\tau \sim \F_q^b$ it could be the case that $\cM_{n,\tau}$ is empty just because, for some $\vec v \in \F_q^b$, $\tau(\vec v) \notin \{0,1/n,2/n,\dots,(n-1)/n,1\}$. For such $\tau$, we can define $\cM_{n,\tau}$ to consist of those matrices which contain either $\lfloor n \cdot \tau(\vec v)\rceil$ or $\lceil n \cdot \tau(\vec v) \rceil$ copies of $\vec v$. As we are always dealing with the setting where $n$ is assumed to be sufficiently large compared to all other parameters, this does not affect the analysis. Hence, we may safely ignore this technicality, which we do for the clarity of exposition.
\end{remark}

Our target is an understanding of the threshold rate for a combinatorial property of random linear codes. The combinatorial properties that we will study are those that are defined by excluding a set of types, as follows.

\begin{definition} [$\tau$-freeness, local properties]
	Given a code $\cC$ and a type $\tau$, we say that $\cC$ is \emph{$\tau$-free} if $\cC$ does not contain any matrix $M \in \cM_{n,\tau}$, i.e., no matrix $M$ of type $\tau$.
	
	For a set $\cT$ of types, where each $\tau \sim \F_q^b$ for some $b \in \N$, we say that $\cC$ is \emph{$\cT$-free} if it is $\tau$-free for all $\tau \in \cT$. We refer to $\cT$-freeness as a \emph{$b$-local property} of codes.
\end{definition}

For a more in-depth discussion of the definition, we refer the reader to, \cite[Section~2]{mosheiff2020ldpc} or \cite[Chapter~3]{resch2020list}. To provide some intuition, we demonstrate how $(\rho,\ell,L)$-list-recoverability may be described as an $L$-local property. We define $\cT$ to be the set of all types $\tau\sim\F_q^L$ such that for some (correlated) distribution $\nu \sim \binom{\F_q}{\ell}$,
\begin{align} \label{eq:list-rec-type-defn}
	\forall i \in [L], ~~~ \Prop_{(\vec \bu,\bS)\sim(\tau,\nu)}[\bu_i \notin \bS] \leq \rho
\end{align}
and furthemore we require
\[
\forall 1 \leq i < j \leq L,~~~ \Prop_{\vec \bu \sim \tau}[\bu_i\neq\bu_j]>0 \ .
\]
(This second condition amounts to requiring that any matrix of type $\tau$ has distinct columns.) We refer to the collection of all these types as $\cT_{\rho,\ell,L}$.

We now characterize (up to $o(1)$ terms) the threshold rate of a property.
\begin{theorem} [\cite{resch2020list},~Theorem~3.3.9: Thresholds for Random Linear Codes] \label{thm:threshold-char}
	Fix $b \in \N$ and let $\cT$ be a set of $b$-local types. Then the threshold rate for $\cT$-freeness is
	\begin{align} \label{eq:threshold-char}
		1-\max_{\tau \in \cT}\min_{A}\curl{\frac{H_q(A\tau)}{\dim(A\tau)}} \pm o_{n\to\infty}(1) \ ,
	\end{align}
	where the minimum is taken over all surjective linear maps $A:\F_q^b \to \F_q^c$ with $c \leq b$.
\end{theorem}

Let us specialize to the case of $\tau$-freeness for a single type $\tau$. Suppose that $R > 1-\min_{A}\curl{\frac{H_q(A\tau)}{\dim(A\tau)}}$. Theorem~\ref{thm:threshold-char} tells us that it is unlikely that a RLC of rate $R$ is $\tau$-free. Stated differently, we can expect that there is at least one matrix of type $\tau$ contained in such an RLC. In fact, while we do not prove this, it is in fact likely that there will be \emph{many} such matrices. For this reason, we use the following terminology for types $\tau$ satisfying $R > 1-\min_{A}\curl{\frac{H_q(A\tau)}{\dim(A\tau)}}$: we call them \emph{abundant}.

In proving an upper bound $R_{\textrm{upper}}$ on the threshold rate for a property of interest (e.g., $(\rho,\ell,L)$-list-recovery), we will follow the following steps. First, we define an appropriate type $\tau$ and prove that a code satisfies the property of interest only if it is $\tau$-free. Informally, we refer to this as a proof that $\tau$ is \emph{bad} for the property of interest. Next, we show that for RLCs of rate $R_{\textrm{upper}}$, the type $\tau$ is abundant. This is the more challenging part of the theorem, as the minimization over the set of all linear maps $A$ is quite challenging to control. Nonetheless, we are able to carry out this program for $(\rho,\ell,L)$-list-recovery, as advertised. 

In proving a lower bound on $R_{\textrm{low}}$ on the threshold rate for a property of interest (e.g., $(\rho,3)$-list-decoding), we need to consider \emph{any} type that is bad for list-decoding, and then show that it is \emph{implicitly rare}: that is, for some matrix $A$, there are relatively few matrices of type $A\tau$, and hence it is likely no matrix of that type lies in the RLC. That is, we must upper bound the ratio of the entropy of $A\tau$ with the dimension of $A\tau$. Here, we have the freedom to choose $A$, but the argument must apply to all types $\tau$. This is especially tricky when given a type $\tau$ whose support is contained in a strict subspace, as then the bound on the entropy must be commensurately smaller. It is for this reason that we only consider small values of $L$, as one suffers from a combinatorial explosion in the number of possible support spaces for the types. 

\paragraph{Thresholds for Random Codes.} For thresholds of random codes, the characterization theorem is simpler in the sense that we do not have to minimize over compressive mappings, at least if the property satisfies certain technical conditions. Fortunately, the characterization applies to list-recoverability, and hence also list-decodability. 

\begin{theorem} [\cite{guruswami2021threshold},~Theorem 2: Thresholds for Random Codes] \label{thm:threshold-char-rc}
	Let $b \in \N$ and let $\cT$ be a set of $b$-local types. Let $T$ be a convex approximation for $\cT$. Then the threshold rate for $\cT$-freeness is 
	\[
	1 - \frac{\max_{\tau \in T}H_q(\tau)}{b} \ .
	\] 
\end{theorem}

\begin{proposition} [\cite{guruswami2021threshold},~Lemma 1] \label{prop:list-rec-convex}
	$\cT_{\rho,\ell,L}$ is a convex approximation for the property of $(\rho,\ell,L)$-list-recoverability. 
\end{proposition}

\section{Lower Bound on List-Size for List-Recovery} \label{sec:list-rec}
Througout this section, the following notations are fixed:
\begin{itemize}
	\item $q \in \N$ is a (fixed) prime power;\footnote{When we discuss random codes, $q$ may be any positive integer.}
	\item $\ell \in \N$ satisfies $1 \leq \ell < q$;
	\item $\rho \in \R$ satisfies $0 < \rho < 1-\frac{\ell}{q}$; and
	\item $\delta>0$ is a small constant.
\end{itemize}
All these parameters are constants, independent of the growing parameter $n$. Our main result in this section is the following theorem.

\begin{theorem} \label{thm:main-list-rec-lower-bound}
	There exists $\eps_{q,\ell,\rho,\delta}>0$ such that for all $0 < \eps < \eps_{q,\ell,\rho,\delta}$ and $n$ sufficiently large, a random linear code in $\F_q^n$ of rate $1-h_{q,\ell}(\rho)-\eps$ is \emph{not} $\parens{\rho,\ell,\lfloor\frac{\log_q\binom{q}{\ell}-(1-h_{q,\ell}(\rho))}{\eps}-\delta\rfloor}$-list-recoverable with probability $1-o(1)$.
\end{theorem}

The proof of this theorem follows the same outline as has been used in, e.g., \cite{GuruswamiLMRSW20}. Namely, we begin by defining a $L$-local type which we show is \emph{bad} for $(\rho,\ell,L)$-list-recovery. Later, we prove that the type is indeed \emph{abundant}, which is the more challenging part of the theorem.

The bad $L$-local type is defined as follows.

\begin{definition} [The bad type for $(\rho,\ell,L)$-list-recoverability] \label{defn:bad-type-list-rec}
	Fix $L\in\N$. Define the distribution $\tau\sim\F_q^L$ via the following procedure for sampling a random vector $\vec \bu = (\bu_1,\dots,\bu_L)$:
	\begin{itemize}
		\item First, $\bS \sim \binom{\F_q}{\ell}$ is sampled uniformly at random;
		\item Second, for $i=1,\dots,L$, we sample $\bu_i \sim \F_q$ as
		\[
		\Pr[\bu_i = x|\bS = S] = \begin{cases}
			\frac{1-\rho}{\ell} & \text{if } x\in S \\
			\frac{\rho}{q-\ell} & \text{if } x\notin S
		\end{cases} \ ,
		\]
		and conditioned on $\bS=S$, the coordinates $\bu_1,\dots,\bu_L$ are \emph{independent}.
	\end{itemize}
\end{definition}

Note that such a type does indeed lie in the set $\cT_{\rho,\ell,L}$. Indeed, if $\nu \sim \bS$ we clearly have
\[
	\forall i \in [L], ~~ \Pr_{(\vec \bu,\bS)\sim(\tau,\nu)}[\bu_i \notin \bS] = \rho
\]
and we also readily have $\Pr_{\vec \bu\sim \tau}[\bu_i \neq \bu_j]>0$.
From~\cite{guruswami2021threshold}, we conclude that $\tau$ is bad for $(\rho,\ell,L)$-list-recovery.

We now claim that the type $\tau$ is indeed abundant, i.e., that it has sufficiently large (relative) entropy. This is the more technical part of the proof, and its proof is deferred to \cref{sec:proof-list-rec-abundance}.

\begin{restatable}{lemma}{listRecAbundance} \label{lemma:list-rec-abundance}
	There exists an integer $L_{\rho,q,\ell,\delta}$ such that for all integers $L \geq L_{\rho,q,\ell,\delta}$, the following holds. Let $\vec \bu \sim \tau$, and let $A \in \F_q^{L' \times L}$ with $L' \leq L$ and $\mathrm{rank}(A)=L'$. Then
	\[
		H_q(A\vec \bu) \geq L' \cdot h_{q,\ell}(\rho) + \log_q\binom{q}{\ell} - 1 + h_{q,\ell}(\rho) - \delta \geq L' \cdot \parens{h_{q,\ell}(\rho) + \frac{\log_q\binom{q}{\ell}-1+h_{q,\ell}(\rho)-\delta}{L}}\ .
	\]
\end{restatable}

Assuming Lemma~\ref{lemma:list-rec-abundance}, we now show that this does indeed yield our target Theorem~\ref{thm:main-list-rec-lower-bound}.

\begin{proof}[Proof of Theorem~\ref{thm:main-list-rec-lower-bound}]
	Let $L_{\rho,q,\ell,\delta/2}$ be the promised constant from Lemma~\ref{lemma:list-rec-abundance}, and choose $\eps_{q,\ell,\rho,\delta} := \frac{\log_q\binom{q}{\ell}-1+h_{q,\ell}(\rho)}{L_{\rho,q,\ell,\delta/2}+1}$. Let $\eps<\eps_{q,\ell,\rho,\delta}$. Let $L = \big\lfloor \frac{\log_q\binom{q}{\ell} - 1 + h_{q,\ell}(\rho)}{\eps}-\delta\big\rfloor$, and define $\tau$ as in Definition~\ref{defn:bad-type-list-rec} with this choice of $L$.
	
	By Lemma~\ref{lemma:list-rec-abundance}, as $L \geq L_{\rho,q,\ell,\delta/2}$ we have that for all surjective linear maps $A: \F_q^L \to \F_q^{L'}$
	\[
		\frac{H_q(A\tau)}{L'} \geq h_{q,\ell}(\rho) + \frac{\log_q\binom{q}{\ell}-1+h_{q,\ell}(\rho)-\delta/2}{L} \ .
	\]
	We note further that as $\tau$ has full support the same is true for $A\tau$, i.e., $\dim(A\tau)=L'$. Thus, by Theorem~\ref{thm:threshold-char} we have that the threshold rate for $\tau$-freeness is at most
	\[
		1- h_{q,\ell}(\rho) - \frac{\log_q\binom{q}{\ell}-1+h_{q,\ell}(\rho)-\delta/2}{L} - o_{n \to \infty} < 1-h_{q,\ell}(\rho)-\eps \ ,
	\]
	where the last inequality holds for large enough $n$. In other words, a random linear code of rate $1-h_{q,\ell}(\rho)-\eps$ contains a matrix $M \in \cM_{n,\tau}$ with probability $1-o(1)$. As we know that a code $\cC$ which contains a matrix of type $\tau$ is not $(\rho,\ell,L)$-list-recoverable, our theorem is proved.
\end{proof}

\subsection{Proof of Lemma~\ref{lemma:list-rec-abundance}} \label{sec:proof-list-rec-abundance}

In this section we prove Lemma~\ref{lemma:list-rec-abundance}.
%
%



\begin{proof}
	Observe that the second inequality is trivial (it just uses that $L \geq L'$), so we focus on the first one. 
	
	First, note that by definition for any $i \in [L]$ we have
	\[
	H_q(\bu_i|\bS) = h_{q,\ell}(\rho) \ .
	\]
	On the other hand,
	\[
	H_q(\bu_i) = 1
	\]
	as $\bu_i$ is uniformly distributed over the randomness of $\bS$.
	Note that if $B \in \F_q^{L' \times L'}$ and $C \in \F_q^{L \times L}$ are any full-rank matrices then $H_q(A\vec \bu) = H_q(BAC\vec \bu)$, so without loss of generality we may apply row operations and column permutations to $A$ so that it has the form
	\[
	\begin{bmatrix}
		& & & & \vline & \vline & \vline & \cdots & \vline  \\
		& & I_{L'} & & \vline & \vec w^{(1)} & \vec w^{(2)} & \cdots & \vec w^{(k)} \\
		& & & & \vline & \vline & \vline & \cdots & \vline \\
	\end{bmatrix} \ ,
	\]
	where $k := L-L'$. When $A$ has this form, we can write
	\[
	A\vec\bu = \begin{bmatrix}
		\bu_{1} \\ \vdots \\ \bu_{L'} \end{bmatrix} + \sum_{i=1}^{k}\ba_{i} \begin{bmatrix} \vline \\ \vec w^{(i)}\\ \vline \end{bmatrix}
	\]
	where $\ba_i = \bu_{L+i}$ for $i \in [k]$.
	
	Define $J_i = \supp(\vec w^{(i)}) \setminus \bigcup_{j=1}^{i-1}J_j$ for $i=1,2,\dots,k$ and $J_{k+1}=[L']\setminus\bigcup_{i=1}^k J_i$, i.e.,
	the sets $J_1,J_2,\dots,J_{k+1}$ form a partition of $[L']$. Note that some of the sets $J_i$ could be empty.
	We emphasize that if $i \notin \supp(\vec w^{(1)}) \cup \dots \cup \supp(\vec w^{(k)})$, then $i \in J_{k+1}$.
	Define $\vec\bu_{J_i}=(\bu_j)_{j\in J_i}$ and $\vec w^{(j)}_{J_i}$ to be the vector $\vec w^{{j}}$ restricted to the
	indexes belonging to $J_i$. Thus, by definition, each component of $\vec w^{(i)}_{J_i}$ is nonzero. We also set $\ba_{k+1}=0$, i.e., we define $\ba_{k+1} \in \F_q$ to be a random variable of $\F_q$ which is equal to $0$ with probability $1$.
	
	For intuition, consider computing the entropy of the random variable $A\vec \bu \in \F_q^{L'}$ by revealing the coordinates of $J_1$, then the coordinates of $J_2$, and so on. Everytime we reveal the coordinates of a new set $J_i$ it will depend on a ``fresh'' coordinate $\bu_{L'+i}$ of $\vec \bu$, which did not influence $(A\vec \bu)_{J_1 \sqcup \dots \sqcup J_{i-1}}$. Thus, there is ``new entropy'' which we can lower bound, permitting us to incrementally lower bound the entropy of $A\vec\bu$.
	
	We now make the following claim. It allows us to conclude that, for coordinates in one of the $J_i$'s with $i<k+1$, the (marginal) entropy of the coordinate is strictly greater than $h_{q,\ell}(\rho)$ (after conditioning on $\bS$).
	
	\begin{restatable}{claim}{sumIncreasesEntropy} \label{claim:sum-increases-entropy}
		For any integers $1 \leq i \leq L'$ and $1 \leq j \leq k$ and $\beta \in \F_q^\times$, we have
		\[
		H_q(\bu_i+\beta\cdot\ba_{j}|\bS) = \lambda h_{q,\ell}(\rho)
		\]
		for some $\lambda = \lambda_{\rho,q,\ell}>1$.
	\end{restatable}
	
	To not distract from the flow of the proof, we defer the proof of Claim~\ref{claim:sum-increases-entropy} to Appendix~\ref{sec:proof-of-sum-increases-entropy-claim}. We now split the proof into two cases, depending on the maximum size of the sets $J_1,\ldots,J_{k+1}$.
	
	\paragraph{Case 1: $\max_{i\in [k+1]}\{|J_i|\}\leq \frac{L(\lambda-1)h_{q,\ell}(\rho)}{\ell}$.} In this case, we do not expect any of the $(A \vec \bu)_{J_i}$'s to have particularly large entropy. So we can lower bound the entropy of $(A \vec \bu)$ ``step-by-step'', lower bounding the additional entropy after revealing each of the $(A\vec \bu)_{J_i}$'s one at a time. Claim~\ref{claim:sum-increases-entropy} allows us to guarantee that we have a sufficiently large increase in entropy.
	
	We begin by applying the chain rule and the definition of mutual information to expand $H_q(A\vec \bu)$ as follows:
	\begin{align*}
		H_q(A\vec \bu) &= H_q(A\vec \bu|\vec \bu_{J_{k+1}},\ba_{k+1}) + I_q(A\vec \bu; \vec \bv_{J_{k+1}},\ba_{k+1})\\
		&= H_q(A\vec \bu|\vec \bu_{J_k},\vec \bu_{J_{k+1}},\ba_k,\ba_{k+1}) + I_q(A\vec \bu; \vec \bu_{J_k},\ba_k|\vec \bu_{J_{k+1}},\ba_{k+1}) + I_q(A\vec \bu; \vec \bu_{J_{k+1}},\ba_{k+1}) \ .
	\end{align*}
	Iterating this argument, one finds
	\begin{align} \label{eq:expansion-of-entropy}
		H_q(A\vec \bu) &= H_q(A\vec \bu|\vec \bu_{J_1},\dots,\vec \bu_{J_{k+1}},\ba_1,\dots,\ba_{k+1})\\
		&+\sum_{i=1}^{k} I_q(A\vec \bu;\vec \bu_{J_i},\ba_{i}|\vec \bu_{J_{i+1}},\dots,\vec \bu_{J_{k+1}},\ba_{i+1},\dots,\ba_{k})+I_q(A\vec\bu;\vec \bu_{j_{k+1}}) \nonumber \\
		&= \sum_{i=1}^{k} I_q(A\vec \bu;\vec \bu_{J_i},\ba_{i}|\vec \bu_{J_{i+1}},\dots,\vec \bu_{J_{k+1}},\ba_{i+1},\dots,\ba_{k})+I_q(A\vec\bu;\vec \bu_{j_{k+1}}) \ ,
	\end{align}
	where in the final equality we used the fact that $J_1,\dots,J_{k+1}$ form a partition of $[L]$ and hence $\vec\bu_{J_1},\dots,\vec \bu_{J_{k+1}},$  $\ba_1,\dots,\ba_{k+1}$ determine $A\vec \bu$. Now, we manipulate a bit the mutual information terms in the above summation. For any $1\leq i \leq k$, we have the term
	\begin{align}\label{eq:entropy-to-lower-bound}
		I_q(A\vec \bu;\vec \bu_{J_i},\ba_{i}&|\vec \bu_{J_{i+1}},\dots,\vec \bu_{J_{k+1}},\ba_{i+1},\dots,\ba_{k}) \nonumber \\
		&\geq I_q\left((A\vec \bu)_{J_i};\bu_{J_i},\ba_{i}|\bu_{J_{i+1}},\dots,\bu_{J_{k+1}},\ba_{i+1},\dots,\ba_{k},\bS\right) \nonumber \\
		&= I_q\left(\vec \bu_{J_i} + \sum_{j=i}^{k}\ba_{j}\cdot \vec w^{(j)}_{J_i};\vec \bu_{J_i},\ba_{i}|\vec \bu_{J_{i+1}},\dots,\vec \bu_{J_{k+1}},\ba_{i+1},\dots,\ba_{k},\bS\right) \nonumber \\
		&= I_q\left(\vec \bu_{J_i} + \ba_{i}\cdot \vec w^{(i)}_{J_i}; \vec \bu_{J_i},\ba_{i}|\vec \bu_{J_{i+1}},\dots,\vec \bu_{J_{k+1}},\ba_{i+1},\dots,\ba_{k+1},\bS\right) \nonumber \\
		&= I_q\left(\vec \bu_{J_i} + \ba_{i}\cdot \vec w^{(i)}_{J_i}; \vec \bu_{J_i},\ba_{i}|\bS\right) \nonumber \\
		&= H_q\left(\vec \bu_{J_i} + \ba_{i}\cdot \vec w_{J_i}^{(i)}|\bS\right)
	\end{align}
	When $i=k+1$, we wish to lower bound the term
	\begin{align}\label{eq:entropy-to-lower-bound-k+1}
		I_q(A\vec \bu;\vec \bu_{J_{k+1}})\geq I_q\left((A\vec \bu)_{J_{k+1}}; \vec{\bu}_{J_{k+1}}\right)=H_q(\vec{\bu}_{J_{k+1}}) \ .
	\end{align}
	
	Consider first $i \in [k]$, i.e., $i<k+1$. If $J_i = \emptyset$, then $H_q\left(\vec \bu_{J_i} + \ba_{i}\cdot \vec w_{J_i}^{(i+1)}|\bS\right) = 0$.\footnote{We interpret both $\vec \bu_{J_i}$ and $\vec w^{(i)}_{J_i}$ as the empty string.} Otherwise, let $d = |J_i| \geq 1$. For convenience, we relabel the random vector $\vec \bu_{J_i} + \ba_i \cdot \vec w^{(i)}_{J_i}$ as
	\begin{align} \label{eq:xyz-vector}
		(\bx_1+y_1\bz, \bx_2+y_2\bz,\dots,\bx_d+y_d\bz)
	\end{align}
	where $y_1,\dots,y_d$ are fixed nonzero elements of $\F_q$, $\bx_1,\dots,\bx_d$ are, conditioned on $\bS$, mutually independent random variables satisfying
	\begin{align} \label{eq:x-vector}
		\Pr[\bx_i=x|\bS=S] = \begin{cases}
			\frac{1-\rho}{\ell} & \text{if } x \in S\\
			\frac{\rho}{q-\ell} & \text{if } x \notin S\\
		\end{cases} \ ,
	\end{align}
	and $\bz$ is sampled as the other $\bx_i$'s.
	
	Recall that, in this case, we are assuming $d\leq \frac{L(\lambda-1)h_{q,\ell}(\rho)}{\ell}$. From \eqref{eq:entropy-to-lower-bound},
	\begin{align*}
		&H_q(\bx_1+y_1\bz, \bx_2+y_2\bz,\dots,\bx_d+y_d\bz|\bS)\\
		&=H_q(\bx_2+y_2\bz,\dots,\bx_d+y_d\bz|\bx_1+y_1\bz,\bS) + H_q(\bx_1+y_1\bz|\bS) \\
		&\geq H_q(\bx_2+y_2\bz,\dots,\bx_d+y_d\bz|\bx_1,\bz,\bS)+H_q(\bx_1+y_1\bz|\bS)\\
		&=\sum_{i=2}^d H_q(\bx_i|\bS) + H_q(\bx_1+y_1\bz|\bS)\\
		&= (d+\lambda-1)h_{q,\ell}(\rho)\geq d\left(h_{q,\ell}(\rho)+\frac{\ell}{L}\right) \ .
	\end{align*}
	We now proceed to lower bound \eqref{eq:entropy-to-lower-bound-k+1}, i.e., the entropy $H_q(\bu_{J_{k+1}})$. For convenience, relabel the random vector $\bu_{J_{k+1}}$ as $(\bx_1, \bx_2,\dots,\bx_d)$. Then,
	\begin{align*}
		H_q(\bx_1, \bx_2,\dots,\bx_d) &=H_q(\bx_1, \bx_2,\dots,\bx_d|\bS)+I_q(\bx_1, \bx_2,\dots,\bx_d;\bS)\\
		&=\sum_{i=1}^d H_q(\bx_i|\bS)+I_q(\bx_1, \bx_2,\dots,\bx_d;\bS)\\
		&\geq dh_{q,\ell}(\rho)+I_q(\bx_1;\bS)=dh_{q,\ell}(\rho)+H_q(\bx_1)-H_q(\bx_1|\bS)\\
		&=(d-1)h_{q,\ell}(\rho)+1\geq (d+\lambda-1)h_{q,\ell}(\rho)\geq d\left(h_{q,\ell}(\rho)+\frac{\ell}{L}\right) \ .
	\end{align*}
	Thus, we have
	\begin{align*}
		H_q(A\vec \bu) &\geq \sum_{i=1}^{k}H_q\parens{\vec \bu_{J_i}+\ba_{i}\cdot \vec w_{J_i}^{(i)}|\bS}+H_q(\bu_{J_{k+1}}) \\
		&\geq \sum_{i=1}^{k+1}|J_i|\left(h_{q,\ell}(\rho)+\frac{\ell}{L}\right) = L'\left(h_{q,\ell}(\rho) +\frac{\ell}{L}\right)
	\end{align*}
	as desired.
	
	\paragraph{Case $2$: $\max_{i\in [k+1]}\{|J_i|\}>\frac{L(\lambda-1)h_{q,\ell}(\rho)}{\ell}$.} For some $d_{\rho,q,\ell,\delta}$ to be chosen later, if we require $L\geq \frac{d_{\rho,q,\ell,\delta}\ell}{(\lambda-1)h_{q,\ell}(\rho)}$, this implies that there exists some $i \in [k+1]$ with $|J_i|> d_{\rho,q,\ell,\delta}$. We will show that the entropy in these coordinates already guarantees that we have a sufficiently large increase in the entropy, even when we use a relatively simple lower bound on the entropy of the other parts. Assuming $i=1$ (which is almost without loss of generality), we do this by demonstrating that $\bu_{J_1}+\ba_1\cdot w_{J_1}^{(1)}$ is informative enough to let us guess the set $\bS$ with very good probability. Fano's inequality (Theorem~\ref{thm:fano}) implies that $\bu_{J_1}+\ba_1\cdot w_{J_1}^{(1)}$ has large entropy. The details follow.
	
	It is useful to consider two subcases.
	\paragraph{Subcase $1$: $i\neq k+1$.} To ease notation, we may reorder indices so that $i=1$.
	
	Analogously to equation~\eqref{eq:expansion-of-entropy} (but now expanding in the opposite direction), we have
	\begin{align} \label{eq:expansion-of-entropy-v2}
		H_q(A\vec \bu) \geq \sum_{i=1}^{k+1}I_q(\vec \bu_{J_i}+\ba_{i}\cdot \vec w^{(i)}_{J_i};\vec \bu_{J_i},\ba_{i}|\vec \bu_{J_{i-1}},\dots,\vec \bu_{J_{1}},\ba_{i-1},\dots,\ba_{1}) \ .
	\end{align}
	We begin by studying the terms in the above summation with $i>1$. Observe that for each such $i$, $\vec\bu_{J_i} + \ba_{i} \cdot \vec w^{(i)}_{J_i}$ is conditionally independent of $(\vec \bu_{J_{i-1}},\dots,\vec \bu_{J_{1}},\ba_{i-1},\dots,\ba_{1})$ given $\bS$. That is, we have a Markov chain $\vec\bu_{J_i} + \ba_{i} \cdot \vec w^{(i)}_{J_i} \to \bS \to (\vec \bu_{J_{i-1}},\dots,\vec \bu_{J_{1}},\ba_{i-1},\dots,\ba_{1})$. The data-processing inequality thus implies that
	\[
		I_q(\vec \bu_{J_i}+\ba_{i}\cdot \vec w^{(i)}_{J_i};\bS) \geq I_q(\vec \bu_{J_i}+\ba_{i}\cdot \vec w^{(i)}_{J_i};\vec \bu_{J_{i-1}},\dots,\vec \bu_{J_{1}},\ba_{i-1},\dots,\ba_{1}) \ .
	\]
	Thus,
	\begin{align} \label{eq:lower-bound-i>1}
		I_q(\vec \bu_{J_i}+\ba_{i}\cdot \vec w^{(i)}_{J_i};&\vec \bu_{J_i},\ba_{i}|\vec \bu_{J_{i-1}},\dots,\vec \bu_{J_{1}},\ba_{i-1},\dots,\ba_{1}) \nonumber \\
		&= H_q(\vec \bu_{J_i}+ \ba_{i}\cdot \vec w^{(i)}_{J_i}|\vec \bu_{J_{i-1}},\dots,\vec \bu_{J_{1}},\ba_{i-1},\dots,\ba_{1}) \nonumber \\
		&~~~~ - H_q(\vec \bu_{J_i}+ \ba_{i}\cdot \vec w^{(i)}_{J_i}|\vec \bu_{J_{i-1}},\dots,\vec \bu_{J_{1}},\ba_{i-1},\dots,\ba_{1},\vec \bu_{J_i},\ba_i) \nonumber \\
		&= H_q(\vec \bu_{J_i}+ \ba_{i}\cdot \vec w^{(i)}_{J_i}|\vec \bu_{J_{i-1}},\dots,\vec \bu_{J_{1}},\ba_{i-1},\dots,\ba_{1}) \nonumber \\
		&= -I_q(\vec \bu_{J_i}+\ba_{i}\cdot \vec w^{(i)}_{J_i};\vec \bu_{J_{i-1}},\dots,\vec \bu_{J_{1}},\ba_{i-1},\dots,\ba_{1}) + H_q(\vec \bu_{J_i} + \ba_{i} \cdot \vec w_{J_i}^{(i)}) \nonumber \\
		&\geq - I_q(\vec \bu_{J_i}+\ba_{i}\cdot \vec w^{(i)}_{J_i};\bS) + H_q(\vec \bu_{J_i}+\ba_{i}\cdot \vec w^{(i)}_{J_i}) \quad \text{(Data-Processing Inequality)} \nonumber \\
		&= H_q(\vec \bu_{J_i}+\ba_{i}\cdot \vec w_{J_i}^{(i)}|\bS) \nonumber \\
		&\geq |J_i| \cdot h_{q,\ell}(\rho) \ .
	\end{align}
	We now consider the $i=1$ term of \eqref{eq:expansion-of-entropy-v2}, which is
	\[
		I_q(\vec \bu_{J_1} + \ba_1 \cdot w_{J_1}^{(1)};\vec \bu_{J_1},\ba_1) = H_q(\bu_{J_1}+\ba_1\cdot w_{J_1}^{(1)}) \ ,
	\]
	and seek an effective lower bound. This again corresponds to lower bounding $H_q(\bx_1 + y_1\bz,\dots,\bx_d+y_d\bz)$, where the $\bx_i$'s, $y_i$'s and $\bz$ are defined as in equations~\eqref{eq:xyz-vector}, \eqref{eq:x-vector} and the surrounding text. Recall that we are assuming that $d \geq d_{\rho,q,\ell,\delta}$. We have
	\begin{align}
		H_q(\bx_1 + y_1\bz,\dots,\bx_d+y_d\bz) &\geq H_q(\bx_1 + y_1\bz,\dots,\bx_d+y_d\bz|\bz) \quad \text{(Conditioning cannot increase entropy)} \nonumber \\
		&= H_q(\bx_1,\bx_2,\dots,\bx_d|\bz) \label{eq:step-to-check} \\
		&= H_q(\bx_1,\bx_2,\dots,\bx_d|\bz,\bS)+I_q(\bx_1,\bx_2,\dots,\bx_d; \bS|\bz) \nonumber \\
		&= H_q(\bx_1,\bx_2,\dots,\bx_d|\bz,\bS) + H_q(\bS|\bz) - H_q(\bS|\bx_1,\bx_2,\dots,\bx_d,\bz) \nonumber \\
		&\geq H_q(\bx_1,\bx_2,\dots,\bx_d|\bz,\bS) + H_q(\bS|\bz) - H_q(\bS|\bx_1,\bx_2,\dots,\bx_d) \label{eq:step-to-continue-from} \ .
	\end{align}
	The equality in \eqref{eq:step-to-check} uses the fact that once $\bz$ is revealed $(\bx_1,\bx_2,\dots,\bx_d)$ and $(\bx_1+y_1\bz,\bx_2+y_2\bz,\dots,\bx_d+y_d\bz)$ have the same entropy. Formally:
	\begin{align*}
		H_q(\vec \bx + \bz \vec y|\bz) = \Eop_{z \sim \bz}\brk{H_q(\vec \bx + \bz\vec y|\bz=z)} = \Eop_{z \sim \bz}\brk{H_q(\vec \bx + z\vec y)} = \Eop_{z \sim \bz}\brk{H_q(\vec\bx)} = H_q(\vec \bx) \ .
	\end{align*}
	We lower bound the first term of \eqref{eq:step-to-continue-from}.
	\[
		H_q(\bx_1,\bx_2,\dots,\bx_d|\bz,\bS) = H_q(\bx_1,\bx_2,\dots,\bx_d|\bS) = dh_{q,\ell}(\rho) \ ,
	\]
	where the first equality uses the fact that $\bz$ is conditionally independent of $\bx_1,\dots,\bx_d$, given $\bS$. The second equality uses the fact that the $\bx_i$'s are mutually conditionally independent given $\bS$, and each satisfies $H_q(\bx_i|\bS) = h_{q,\ell}(\rho)$.
	
	Next, we look at the $H_q(\bS|\bz)$ term of \eqref{eq:step-to-continue-from}. Recalling the distribution of $\bz$ (it is one of the $\bu_i$'s, relabeled) we may apply Bayes' Rule for conditional entropy to get
	\[
		H_q(\bS|\bz) = H_q(\bz|\bS) - H_q(\bz) + H_q(\bS) = h_{q,\ell}(\rho) - 1 + \log_q\binom{q}{\ell} \ .
	\]
	Thus, we have
	\[
		H_q(\bx_1+y_1\bz,\bx_2+y_2\bz,\dots,\bx_d+y_d\bz) \geq dh_{q,\ell}(\rho) + \log_q\binom{q}{\ell} - 1 + h_{q,\ell}(\rho) - H_q(\bS|\bx_1,\dots,\bx_d) \ .
	\]
	
	It thus remains to upper bound $H_q(\bS|\bx_1,\bx_2,\dots,\bx_d)$, a task for which we use Fano's inequality, \cref{thm:fano}. In order to do this, we must find a function $f: \F_q^d\rightarrow \binom{\F_q}{\ell}$ so that $\perr = \Pr[f(\bx_1,\ldots,\bx_d)\neq \bS]$ is very small. We define $f$ in the most obvious way: $f(x_1,\ldots,x_d):=\{y_1,\ldots,y_\ell\}$ if $y_1,\ldots,y_\ell$ are the $\ell$ most frequent elements appearing in $(x_1,\ldots,x_d)$ (breaking ties arbitrarily). For any $\alpha\in \F_q$, let $\bc_\alpha=|\{\bx_i=\alpha: i\in [d]\}|$ be the random variable counting the number of $\bx_i$'s taking on the value $\alpha$. Observe that
	\begin{align*}
		\E[\bc_\alpha|\bS=S] = \begin{cases}
			\frac{d(1-\rho)}{\ell} & \text{if } \alpha \in S \\
			\frac{d\rho}{q-\ell} & \text{if }\alpha \notin S
		\end{cases}
	\end{align*}
	Note that the assumption $\rho < 1-\frac{\ell}{q}$ is equivalent to $\frac{1-\rho}{\ell} > \frac{\rho}{q-\ell}$. By the Chernoff bound, we therefore have that for any $S \in \binom{\F_q}{\ell}$, $\alpha \in S$ and $\beta \notin S$:
	\[
		\Pr[\bc_\alpha < \bc_\beta|\bS=S] \leq \exp\parens{-\Omega_{q,\ell,\rho}(d)} \ .
	\]
	Thus, by applying the total probability rule and taking a union bound over all pairs $(\alpha,\beta) \in S\times(\F_q\setminus S)$, we may upper bound the probability of error $\perr$ as
	\begin{align*}
		\perr &= \Pr[f(\bx_1,\dots,\bx_d) \neq \bS] = \frac{1}{\binom{q}{\ell}}\sum_{S \in \binom{\F_q}{\ell}}\Pr[f(\bx_1,\dots,\bx_d)\neq S|\bS=S] \\
		&\leq \frac{1}{\binom{q}{\ell}}\sum_{S \in \binom{\F_q}{\ell}} \ell(q-\ell)\exp\parens{-\Omega_{q,\ell,\rho}(d)} \leq \exp\parens{-\Omega_{q,\ell,\rho}(d)} \ .
	\end{align*}
	Fano's inequality therefore yields
	\[
		H_q(\bS|\bx_1,\ldots,\bx_d)\leq H_q(p_{err})+p_{err}\log_2 \binom{q}{\ell}\leq \exp(-\Omega_{q,\ell,\rho}(d)) \ .
	\]
	Thus, for any $\delta>0$, there exists $d_{\rho,q,\ell,\delta}$ such that if $|J_1| = d \geq d_{\rho,q,\ell,\delta}$ we have $H_q(\bS|\bx_1,\dots,\bx_d) \leq \delta$. Putting everything together:
	\begin{align} \label{eq:lower-bound-on-i=1}
		H_q\left(\vec \bu_{J_{1}} + \ba_{1}\cdot w_{J_{1}}^{(1)}\right) \geq dh_{q,\ell}(\rho) + \log_q\binom{q}{\ell} - 1 + h_{q,\ell}(\rho) - \delta \ .
	\end{align}
	Thus, combining Equations \eqref{eq:lower-bound-i>1} and \eqref{eq:lower-bound-on-i=1}, we obtain the desired lower bound on $H_q(A\vec\bu)$.
	\begin{align*}
		H_q(A\vec\bu) &\geq |J_1|h_{q,\ell}(\rho) + \log_q\binom{q}{\ell} - 1 + h_{q,\ell}(\rho) - \delta + \sum_{i=2}^{k+1} |J_i|h_{q,\ell}(\rho) \\
		&= h_{q,\ell}(\rho) \sum_{i=1}^{k+1}|J_i| + \log_q\binom{q}{\ell} - 1 + h_{q,\ell}(\rho) - \delta \\
		&=L' \cdot h_{q,\ell}(\rho) + \log_q\binom{q}{\ell} - 1 + h_{q,\ell}(\rho) - \delta \ .
	\end{align*}
	
	\paragraph{Subcase $2$: $i=k+1$} The proof follows almost the same as Subcase 1 except that we now want to prove
	\[
		H_q(\bx_1,\dots,\bx_d) \geq dh_{q,\ell}(\rho) + \log_q\binom{q}{\ell} - \delta
	\]
	as $\ba_{k+1}=0$ in this case, and $\bz$ corresponds to $\ba_{k+1}$. Observe that this entropy equals
	\[
		H_q(\bx_1,\dots,\bx_d|\bS)+H_q(\bS)-H_q(\bS|\bx_1,\ldots,\bx_d) = dh_{q,\ell}(\rho)+\log_q\binom{q}{\ell}-H_q(\bS|\bx_1,\ldots,\bx_d).
	\]
	Apply Fano's inequality in the same manner as in the previous subcase to the last term and we can conclude that
	$H_q(\bS|\bx_1,\ldots,\bx_d)\leq \delta$ when $d\geq d_{\rho,q,\ell,\delta}$.
	
	This completes the proof of this case, and therefore also the proof of the lemma.
\end{proof}

\subsection{List-recoverability lower bound for random codes}

For context, we provide nearly matching upper and lower bounds for list-recovery for \emph{uniformly} random codes. There is a similar result for list-recovery provided in~\cite{guruswami2021threshold}, but it is not optimized for the case of capacity-approaching codes. 

\begin{theorem} \label{thm:list-rec-random}
	There exists $\eps_{q,\ell,\rho,\delta}>0$ such that for all $0 < \eps < \eps_{q,\ell,\rho,\delta}$ and $n$ sufficiently large, a random code in $\F_q^n$ of rate $1-h_{q,\ell}(\rho)-\eps$ is \emph{not} $\parens{\rho,\ell,\lfloor \frac{\log_q\binom{q}{\ell}}{\eps}-\delta\rfloor}$-list-recoverable.
	
	On the other hand, for any $\eps>0$ and $n$ sufficiently large, a random code in $\F_q^n$ of rate $1-h_{q,\ell}(\rho)-\eps$ is $\parens{\rho,\ell,\lceil\frac{\log_q\binom{q}{\ell}}{\eps}+1\rceil}$-list-recoverable.
\end{theorem}

In this way, we can essentially pin-down the list size of a rate $1-h_{q,\ell}(\rho)-\eps$ random code to one of three possible values. This is similar to the result on the list-decodability of binary random linear codes from~\cite{GuruswamiLMRSW20}.

Observe that if we want to prove a lower bound on the threshold rate of a random code instead of a random linear code, we can restrict to the case that the matrix $A$ from Lemma~\ref{lemma:list-rec-abundance} is the identity matrix. Thus, we are in the setting $k=0$ and so we are in Subcase 2 where $i=k+1=1$. We may reuse the lower bound on the entropy $H_q(\bx_1,\dots,\bx_L)$ from this case, yielding the following lemma.
\begin{lemma}\label{lem:list-rec-rand-lb}
	There exists an integer $L_{p,q,\ell,\delta}$ such that for all integers $L \geq L_{\rho,q,\ell,\delta}$, the following holds. Let $\vec \bu \sim \tau$. Then
	\[
		H_q( \vec \bu) \geq L \cdot h_{q,\ell}(\rho) + \log_q\binom{q}{\ell} + h_{q,\ell}(\rho) - \delta \geq L \cdot \parens{h_{q,\ell}(\rho) + \frac{\log_q\binom{q}{\ell}-\delta}{L}}\ .
	\]
\end{lemma}

An argumentation analogous to that of the proof of Theorem~\ref{thm:main-list-rec-lower-bound} yields the following corollary.

\begin{corollary}
	There exists $\eps_{q,\ell,\rho,\delta}>0$ such that for all $0 < \eps < \eps_{q,\ell,\rho,\delta}$ and $n$ sufficiently large, a random code with rate $1-h_{q,\ell}(\rho)-\eps$ is with high probability \emph{not} $(\rho,\ell,\lfloor \frac{\log_q\binom{q}{\ell}}{\eps}-\delta\rfloor)$-list recoverable.
\end{corollary}
We proceed to pin down the threshold rate of list recovery of random code by showing an upper bound.
\begin{lemma} \label{lem:list-rec-rand-ub}
	Let $q$ be a prime power, $1 \leq \ell \leq q$ an integer and $\rho \in \left(0,1-\tfrac\ell q\right)$. A random code with rate
	$1-h_{q,\ell}(\rho)-\frac{\log_q \binom{q}{\ell}}{L}$ is with high probability $(\rho,\ell,L)$-list recoverable.
\end{lemma}

Clearly, the combination of Lemmas~\ref{lem:list-rec-rand-lb}~and~\ref{lem:list-rec-rand-ub} yields our target, \cref{thm:list-rec-random}. 

\begin{proof}
	It suffices to prove an upper bound on $H_q(\tau)$ for any $\tau \in \cT_{\rho,\ell,L}$. In pariticular, this means that for some $\nu \sim \binom{\F_q}{\ell}$ we have 
	\begin{equation}\label{eq:u}
		\forall i\in [L], \Pr_{(\vec \bu,\bS)\sim (\tau,\nu)}[\bu_i\notin \bS]\leq \rho.
	\end{equation}
	Note that
	$$
	H_q(\tau)=H_q(\tau|\tau')+H_q(\tau')-H_q(\tau'|\tau)\leq H_q(\tau|\tau')+H_q(\tau')\leq H_q(\tau|\tau')+\log_q \binom{q}{\ell}.
	$$
	We turn to upper bound $H_q(\tau|\tau')$. Let $(\vec\bu,\bS)\sim (\tau,\tau')$ with $\vec\bu=(\bu_1,\ldots,\bu_{L})$ and we compute
	$$
	H_q(\tau|\tau')=H_q(\bu|\bS)\leq \sum_{i=1}^{L}H_q(\bu_i|\bS)\leq Lh_{q,\ell}(\rho).
	$$
	The last inequality is due to
	$$H_q(\bu_i|\bS)\leq \Pr[\bu_i\in \bS]\log_q \frac{\Pr[\bu_i\in \bS]}{\ell}+\Pr[\bu_i\notin \bS]\log_q \frac{\Pr[\bu_i\notin \bS]}{q-\ell}\leq h_{q,\ell}(\rho).$$
	The proof is completed.
\end{proof}

\section{List-Decoding with Small Lists} \label{sec:list-dec}
In this section, we investigate the list-decodability of random codes and random linear codes with constant list sizes. Specifically, for list-of-$3$ decoding over the binary field, we can show that the threshold rate for list-decoding of random linear codes is strictly better than that for list-decoding uniformly random codes. Further, for larger field sizes we are able to show that the threshold rate for list-of-$2$ decoding is strictly better for random linear codes than for uniformly random codes. This extends the result of \cite{guruswami2021threshold} which only applies to list-of-$2$ decoding for binary codes.

For our lower bound on the threshold rates for RLCs, we follow the following procedure. First, we consider any type that is bad for, e.g., $(\rho,3)$-list-decoding, i.e., a type from $\cT_{\rho,1,3}$. For any such type $\tau$, we upper bound $\frac{H_q(A\tau)}{\dim(A\tau)}$  for the linear map $A$ sending $(x_1,x_2,x_3) \mapsto (x_1-x_3,x_2-x_3)$. This is straightforward when the $\dim(A\tau)$ is full (requiring essentially only the concavity of the entropy function); when it is smaller, more careful reasoning is required. 

As a final contribution, we recall that in \cite{GuruswamiLMRSW20} it is shown that over the binary field the threshold rate for random linear codes is strictly better than random codes in the capacity-approaching regime. We observe that their techniques can be extended to show that such a trend holds for any constant list size $L$ (assuming the decoding radius $\rho$ is not too large). To do this, we first prove a lower bound on the threshold rate of binary random linear codes by applying the argument in \cite{LiW18} and an upper bound on the threshold rate of binary random codes following the argument in \cite{GuruswamiLMRSW20}. Although our proof resorts to known techniques, such results were not stated before and greatly strengthen our belief that random linear codes perform better than random codes. In light of the available evidence, a reasonable conjecture would be that the for all alphabet sizes, the threshold rate of random linear codes is strictly better than that of random codes.

\subsection{List-of-$3$ Decoding for Binary Alphabet}
In this section, we study the threshold rate for list-of-$3$ decoding binary codes. We recall that the Plotkin point for list-of-3 decoding binary codes, i.e., the maximum value of $\rho$ for which $(\rho,4)$-list-decoding with positive rate is possible, is $5/16$~\cite{alon2018list}. Our main theorem is the following:

\begin{theorem}\label{thm:thresholdrandomlinear_2}
	Let $\rho \in (0,5/16)$. The threshold rate for $(\rho,4)$-list-decoding a random linear code over $\F_2$ is at least
	$$1-\max\curl{\frac{H_2(x_1,x_2)+2x_1+x_2\log_2 3}{3}:x_1+2x_2 \leq 4\rho, x_1+x_2 \leq 1, x_1,x_2 \geq 0}.$$
\end{theorem}

\begin{proof}
	Let $\tau \in \cT_{\rho,1,4}$, which we recall means $\tau \sim \F_2^4$ and there is a distribution $\nu \sim \F_2$ for which
	\begin{align} \label{eq:list-of-3-cond}
		\forall i \in [4], ~~\Pr_{(\vec \bu_i,\bz) \sim (\tau,\nu)}[\bu_i \neq \bz] \leq \rho
	\end{align}
	and furthermore
	\[
	\forall 1 \leq i < j \leq 4, ~~ \Pr_{\vec \bu \sim \tau}[\bu_i \neq \bu_j] > 0 \ .
	\]
	Note that condition~\eqref{eq:list-of-3-cond} implies
	\begin{align} \label{eq:list-of-3-avg}
		\sum_{i=1}^4 \Pr_{(\vec \bu,\bz) \sim (\tau,\nu)}[\bu_i \neq \bz] \leq 4\rho \ .
	\end{align}
	Note that if $\bz = \mathrm{MAJ}(\vec \bu)$ then the left-hand-side of \eqref{eq:list-of-3-avg} can only decrease. Thus, we have
	\begin{align} \label{eq:list-of-3-maj}
		\sum_{i=1}^4 \Pr_{(\vec \bu,\bz) \sim (\tau,\nu)}[\bu_i \neq \mathrm{MAJ}(\vec \bu)] \leq 4\rho \ .
	\end{align}
	
	Define the sets $A_0=\{\vec v\in \F_2^4: \wt(\vec v)=0,4\}$, $A_1=\{\vec v\in \F_2^4: \wt(\vec v)=1,3\}$ and $A_2=\{\vec v\in \F_2^4,
	\wt(\vec v)=2\}$.
	It is clear that $|A_0|=2, |A_1|=8,|A_2|=6$. Let $\tau(A_1)=x_1,\tau(A_2)=x_2$ and $\tau(A_0)=1-x_1-x_2$. Observe that \eqref{eq:list-of-3-maj} implies that $x_1+2x_2\leq 4\rho$. We also clearly have $x_1+x_2 \leq 1$ and $x_1,x_2 \geq 0$; in the sequel, these two constraints are always assumed to hold for $x_1,x_2$.
	
	We consider $\tau'=A\tau$ where $A:\F_2^4 \to \F_2^3$ is the linear map defined by $(a,b,c,d) \mapsto (a+d,b+d,c+d)$.
	This implies $\tau'(a,b,c)=\tau(a,b,c,0)+\tau(a+1,b+1,c+1,1)$.
	We note that $(a,b,c,0)$ and $(a+1,b+1,c+1,1)$ belong to the same set $A_i$.
	Therefore,
	$$
	H_2(\tau')=\frac{1}{2}\sum_{\vec v \in \F_2^4} -(\tau(\vec v)+\tau(\vec v+\vec 1))\log_2 (\tau(\vec v)+\tau(\vec v+\vec 1))\leq  H_2(x_1,x_2)+2x_1+x_2\log_2 3,
	$$
	due to the concavity of function $f(x)=x\log_2 x$.
	If $\dim(\tau')=3$, we have
	\begin{equation}\label{eq:thresholdrate}
		\min_{B}\frac{H_2(B\tau)}{\dim(B\tau)}\leq \max_{x_1+2x_2\leq 4\rho}\frac{H_2(x_1,x_2)+2x_1+x_2\log_2 3}{3} \ ,
	\end{equation}
	where the minimization is over all compressing linear maps $B$. 
	
	Otherwise, $\dim(\tau)\leq 3$. If $\dim(\tau)=3$, this implies $(1,1,1,1)$ belongs to the support of $\tau$. There are another two linearly independent vectors $\vec v_1,\vec v_2$ in its support. We note that it suffices to consider the linearly independent vectors
	so as to ensure that the matrix generated by $\tau$ has distinct columns. By symmetry, it suffices to consider vectors of weight $1$ or weight $2$. It is clear that at least one of them must have weight $2$. By symmetry, we assume $\vec v_1=(1100)$. To generate distinct columns, the first component and the second component of $\vec v_2$ must be different and so do the third and fourth component. This implies that $\vec v_2=(0101)$ or $\vec v_2=(1010)$. Due to the symmetry, we only need to consider the case $\vec v_1=(1100)$, $\vec v_2=(0101)$.
	Once the support set of $\tau$ is determined, we find that the support set of $\tau$ is exactly $A_0\cup A_2$.
	A simple calculation shows
	$$
	H_2(\tau')\leq H_2(x_2)+x_2\log_2 3,
	$$
	subject to $x_2\leq 2\rho$. As this bound clearly increases with $x_2$, we have
	\begin{align} \label{eq:worse-thresholdrate}
		\frac{H_2(\tau')}{2} \leq \frac{h_2(2\rho)+2\rho\log_23}{2} \ .
	\end{align}
	To show the upper bound from \eqref{eq:thresholdrate} is indeed larger, one can optimize the equation on the boundary $x_1+2x_2=4\rho$. To do this, one may take a derivative and solve for the critical point, which is a quadratic equation in $x_2$ whose positive root is
	\[
	2(\rho-1) + 2\sqrt{1-2\rho+4\rho^2} \ .
	\]
	A (tedious) computation shows that this bound does dominate $\frac{h_2(2\rho)+2\rho\log_23}{3}$; see Figure~\ref{fig:bound-comp}.
	
	\begin{figure}
		\centering
		\includegraphics{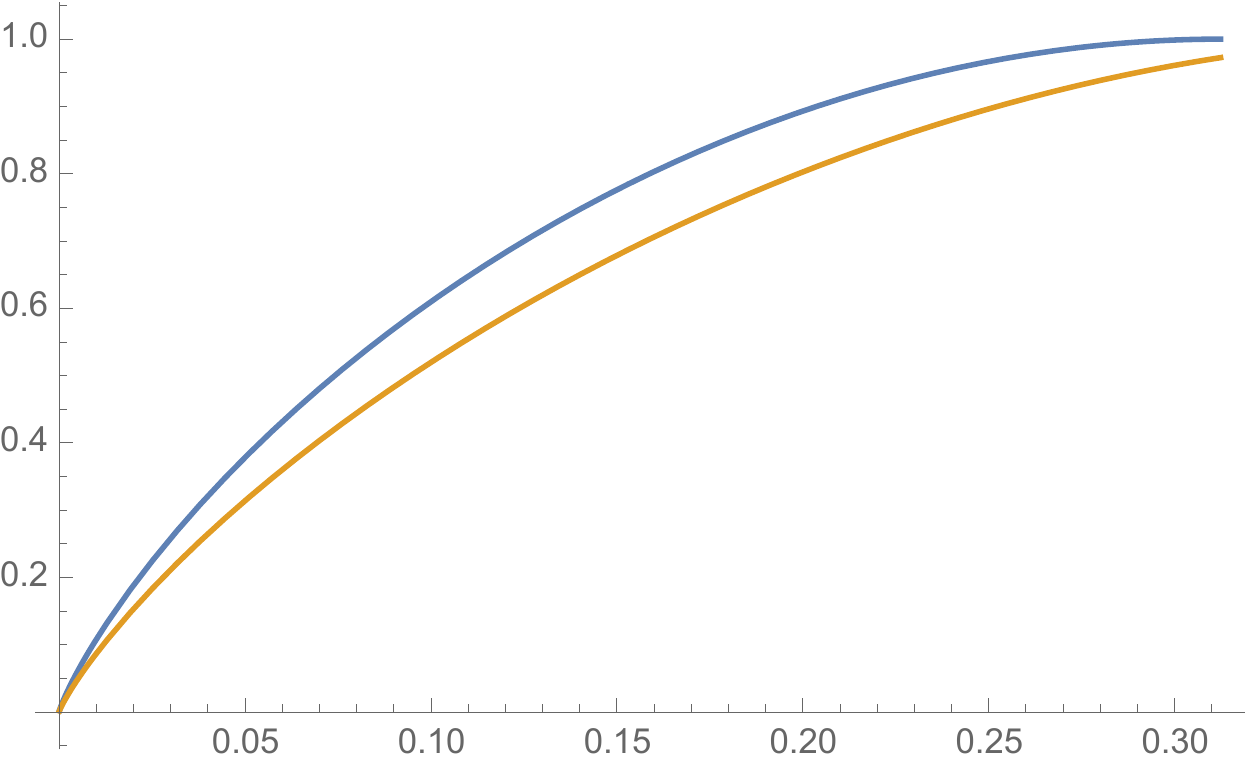}
		\caption{In blue, a (lower bound) for \eqref{eq:thresholdrate} is plotted. In orange, \eqref{eq:worse-thresholdrate} is plotted. We see $\eqref{eq:thresholdrate} \geq \eqref{eq:worse-thresholdrate}$.}\label{fig:bound-comp}
	\end{figure}
	Now, we proceed to the case $\dim(\tau)=2$. In this case, $(1,1,1,1)$ does not belong to the support of $\tau$.
	There are two linearly independent vectors $\vec v_1,\vec v_2$ in its support. By symmetry, the same argument shows that the only case is $\vec v_1=(1100)$, $\vec v_2=(0101)$. We conclude that
	$$
	H_2(\tau')=H_2(\tau)=H_2(0,x_2)+x_2\log_2 3,
	$$
	subject to $x_2\leq 2\rho$. The same conclusion applies. The case $\dim(\tau)=1$ will result in that the matrix generated by $\tau$ does not have distinct columns. We can thus easily rule out this possibility. The proof is completed.
\end{proof}

Next, for context, we consider the threshold rate for $(\rho,4)$-list-decoding uniformly random codes.

\begin{theorem}\label{thm:thresholdrandom_2}
	Let $\rho \in (0,5/16)$. The threshold rate for $(\rho,4)$-list decoding a random code over $\F_2$  is
	$$1-\max \curl{\frac{1+H_2(x_1,x_2)+2x_1+x_2\log_2 3}{4}:x_1+2x_2 \leq 4\rho, x_1+x_2 \leq 1, x_1,x_2 \geq 0}.$$
\end{theorem}

\begin{proof}
	Let $\tau \in \cT_{\rho,1,4}$, and again define the sets $A_0=\{\vec v\in \F_2^4: \wt(\vec v)=0,4\}$, $A_1=\{\vec v\in \F_2^4: \wt(\vec v)=1,3\}$ and $A_2=\{\vec v\in \F_2^4, \wt(\vec v)=2\}$.
	Recall $|A_0|=2, |A_1|=8,|A_2|=6$. Letting $\tau(A_1)=x_1,\tau(A_2)=x_2$ and $\tau(A_0)=1-x_1-x_2$, the same reasoning that we used in the proof of Theorem~\ref{thm:thresholdrandomlinear_2} tells us $x_1+2x_2\leq 4\rho$. Now:
	$$
	H_2(\tau)=\sum_{\vec v \in \F_2^4} -\tau(\vec v)\log_2 (\tau(\vec v))\leq 1+ H_2(x_1,x_2)+2x_1+x_2\log_2 3,
	$$
	due to the concavity of function $f(x)=x\log_2 x$.
	This means the threshold rate of $(\rho,4)$-list decoding a random code over $\F_2$ is at least
	$$
	1-\max_{\tau\in T}\frac{H_2(\tau)}{4}\geq 1-\max_{x_1+2x_2\leq 4\rho} \frac{1+H_2(x_1,x_2)+2x_1+x_2\log_2 3}{4}.
	$$
	On the other hand, let $x_1$ and $x_2$ be the values achieving the maximum of $1+H_2(x_1,x_2)+2x_1+x_2\log_2 3$.
	We construct the distribution $\tau$ such that $\tau(\vec v)=\frac{x_i}{|A_i|}$ for $\vec v\in A_i$, $i=0,1,2$.
	It is easy to verify that such $\tau$ achieves the maximum value and thus this lower bound is indeed the threshold rate for $(\rho,4)$-list decoding a random code. 
\end{proof}

As $\frac{1+F}{4} \geq \frac{F}{3}$ for all $F \leq 3$, the lower bound on the threshold rate provided by Theorem~\ref{thm:thresholdrandomlinear_2} is greater than the exact value from Theorem~\ref{thm:thresholdrandom_2}. This demonstrates that random linear codes do indeed perform better.

\subsection{List-of-$2$ Decoding for Arbitrary Alphabets}

We now study list-of-2 decoding over $\F_q$ for $q \geq 3$. Here, the Plotkin point is to the best of our knowledge unknown, and we just prove our result for $\rho < 1/3$.

\begin{theorem}\label{thm:thresholdrandomlinear}
	Let $\rho \in (0,1/3)$. The threshold rate for $(\rho,3)$-list decoding random linear code over $\F_q$ with $q\geq 3$ is at least
	$$1-\max\curl{\frac{H_q(x_1,x_2)+x_1\log_q 3(q-1)+x_2\log_q (q-1)(q-2)}{2}:x_1+2x_2 \leq 3\rho, x_1+x_2 \leq 1, x_1,x_2 \geq 0}.$$
\end{theorem}

\begin{proof}
	Let $\tau \in \cT_{\rho,1,3}$, which we recall means
	\[
		\forall i \in [3], ~~ \Pr_{(\vec \bu,\bz) \sim (\tau,\nu)}[\bu_i \neq \bz] \leq \rho
	\]
	and furthermore
	\begin{align} \label{eq:distinctness}
		\forall 1 \leq i < j \leq 3, ~~	\Pr_{\vec{\bu}\sim\tau}[\bu_i \neq \bu_j]>0 \ .
	\end{align}
	Let
	$$A_0=\{(x,x,x):x\in \F_q\}, A_2=\{(x,y,z):x\neq y, y\neq z,x\neq z\}\subseteq \F_q^3, A_2=\F_q^3/(A_0\cup A_2).$$
	Assume that $\tau(A_1)=x_1$, $\tau(A_2)=x_2$ and $\tau(A_0)=1-x_1-x_2$.
	Since the linear code is $(\rho,3)$-list decodable, by assuming $\bz = \mathrm{MAJ}(\vec \bu)$ we observe that $x_1+2x_2\leq 3\rho$ (this is analogous to the argumentation from the proof of \cref{thm:thresholdrandomlinear_2}). Clearly, we also have the constraint $x_1+x_2 \leq 1$ and $x_1,x_2 \geq 0$: in the remainder of the proof, these constraints are assumed to be satisfied.
	
	For each distribution $\tau$, we want to find $\tau'=A\tau$ to reach $\min_{\tau'\in I_\tau}\frac{H_q(\tau')}{\dim(\tau')}$.
	If $\dim(\tau)=3$, the same argument in Theorem \ref{thm:thresholdrandom} shows that
	\begin{equation}\label{eq:dim_3}
		\frac{H_q(\tau)}{\dim (\tau)}\leq  \max_{x_1+2x_2\leq 3\rho}\frac{1+H_q(x_1,x_2)+x_1\log_q 3(q-1)+x_2\log_q (q-1)(q-2)}{3} \ .
	\end{equation}
	We now consider $\tau'$ defined by the linear map $(x-z,y-z)$. The kernel of this linear map is $\{(x,x,x):x\in \F_q\}$.
	Therefore, $\tau'(a,b)=\sum_{x\in \F_q} \tau(x+a,x+b,x)$.
	Let $B_0=\{(0,0)\}$, $B_1=\{(0,a),(a,0),(a,a):a\in \F_q^*\}$ and $B_2=\F_q^2/(B_0\cup B_1)$.
	Observe that the preimage of the linear map in $B_i$ is exactly $A_i$, i.e., $\tau'(B_i)=\tau(A_i)$. Thus, we have
	\begin{align*}
		H_q(\tau')&=-\sum_{i=0}^2\sum_{\vec v\in B_i}\tau'(\vec v)\log_q\tau'(\vec v)\\
		&\leq -\tau(A_0)\log_q \tau(A_0)-\tau(A_1)\log_q \frac{\tau(A_1)}{3(q-1)}-\tau(A_2)\log_q \frac{\tau(A_2)}{(q-1)(q-2)}\\
		&=H_q(x_1,x_2)+x_1\log_q 3(q-1)+x_2\log_q (q-1)(q-2).
	\end{align*}
	The first inequality is due to the concavity of $x\log_q x$.
	If $\dim(\tau')=2$, we obtain that
	\begin{equation}\label{eq:tau}
		\frac{H_q(\tau')}{\dim(\tau')}\leq \frac{H_q(x_1,x_2)+x_1\log_q 3(q-1)+x_2\log_q (q-1)(q-2)}{2}.
	\end{equation}
	This is smaller than the upper bound given by \eqref{eq:dim_3} as $\frac{F}{2}\leq\frac{F+1}{3}$ for $F\leq 2$. It remains to consider the case $\dim(\tau')=1$ under this linear map. We divide it into two cases.
	
	\paragraph{Case $1$: $\dim(\tau)=2$.} In this case, the support of $\tau$ must contain a nonzero element $(a,a,a)$ in $A_0$. By the linearity of $A_0$, we assume that $(b,c,d)\notin A_0$ also lies in the support of $\tau$. First, we claim that $b,c,d$ must be distinct. Otherwise, without loss of generality, we assume that $b=c$. Then, the support of $\tau$ is contained in $\spa_{\F_q}\{(a,a,a), (b,b,d)\}\subseteq A_1\cup A_0$. The first two coordinates of $\tau$ are always the same which contradicts the distinctness requirement.
	Thus, the support set of $\tau'$ is contained in $\{\lambda (b-d,c-d):\lambda\in \F_q\}\subseteq B_0\cup B_2$.
	This also implies that $\tau(A_1)=x_1=0$. This leads to
	$$
		\frac{H_q(\tau)}{\dim(\tau)}\leq\frac{1}{2}\bigg( -(1-x_2)\log_q \frac{1-x_2}{q}-x_2\log_q \frac{x_2}{q(q-1)}\bigg)=\frac{1+H_q(0,x_2)+x_2\log_q (q-1)}{2}.
	$$
	and
	$$
		\frac{H_q(\tau')}{\dim(\tau')}\leq -(1-x_2)\log_q(1-x_2)-x_2\log_q \frac{x_2}{q-1}=H_q(0,x_2)+x_2\log_q (q-1) = h_q(x_2).
	$$
	Clearly, the latter upper bound is smaller. Its maximum value is attained at $x_2=\frac{3\rho}{2}$ for $x_2\leq 1-\frac{1}{q}$.
	We conclude that
	\begin{equation}\label{eq:tau1}
		\frac{H_q(\tau')}{\dim(\tau')}\leq H_q(0,3\rho/2)+\frac{3\rho}{2}\log_q (q-1)
	\end{equation}
	
	\paragraph{Case $2$: $\dim(\tau)=1$.} The same argument in Case $1$ implies that the support of $\tau$ must contain an element
	$(x,y,z)$ such that $x,y,z$ are distinct. It is clear that $\tau(A_1)=x_1=0$. The same argument shows that
	$$
	H_q(\tau)\leq (1-x_2)\log_q(1-x_2)-x_2\log_q \frac{x_2}{q-1}=H_q(0,x_2)+x_2\log_q (q-1)
	$$
	subject to $x_2\leq \frac{3p}{2}$. We obtain the same function appearing in Case $1$ and the same conclusion holds.
	
	It remains to compare the upper bound \eqref{eq:tau} with \eqref{eq:tau1}.
	For $\rho<\frac{1}{3}$, if we plug $x_1=3\rho,x_2=0$ into Equation \eqref{eq:tau}, we obtain that
	$$
	\max_{x_1+2x_2\leq 3\rho}\frac{H_q(x_1,x_2)+x_1\log_q 3(q-1)+x_2\log_q (q-1)(q-2)}{2}\geq
	\frac{H_q(3\rho,0)+3\rho\log_q 3(q-1)}{2}.
	$$
	Observe that
	\begin{align*}
		2\bigg(H_q\left(0,\frac{3\rho}{2}\right)+\frac{3\rho}{2}\log_q (q-1)\bigg)&-H_q(3\rho,0)-3\rho\log_q 3(q-1)\\
		&=\frac{1}{\log_2 q}\left(2H_2\left(0,\frac{3\rho}{2}\right)-H_2(3\rho,0)-3\rho\log_2 3\right).
	\end{align*}
	By computer program, one can show that $H_2(0,3\rho/2)-H_2(3\rho,0)-3\rho\log_2 3$ is always negative for $\rho<\frac{1}{3}$.
	This implies that
	$$\max_{\tau}\min_{\tau'\in I_\tau}\frac{H_q(\tau')}{\dim(\tau')}\leq \max_{x_1+2x_2\leq 3\rho}\frac{H_q(x_1,x_2)+x_1\log_q 3(q-1)+x_2\log_q (q-1)(q-2)}{2}.$$
	The proof is completed.
\end{proof}

For context, we again consider random codes.

\begin{theorem}\label{thm:thresholdrandom}
	Let $\rho \in (0,1/3)$. The threshold rate for $(\rho,3)$-list decoding random code over $\F_q$  is
	$$1-\max \curl{\frac{1+H_q(x_1,x_2)+x_1\log_q 3(q-1)+x_2\log_q (q-1)(q-2)}{3} : x_1+2x_2\leq 3\rho, x_1+x_2 \leq 1, x_1,x_2 \geq 0}.$$
\end{theorem}

\begin{proof}
	According to \cref{thm:threshold-char-rc}, the threshold rate is
	$$
	1 - \frac{\max_{\tau \in \cT_{\rho,1,3}}H_q(\tau)}{3}.
	$$
	We first prove an upper bound on $\max_{\tau\in \cT_{\rho,1,3}}H_q(\tau)$.
	Let
	$$A_0=\{(x,x,x):x\in \F_q\}, ~A_2=\{(x,y,z):x\neq y, y\neq z,x\neq z\}\subseteq \F_q^3, ~A_1=\F_q^3/(A_0\cup A_2) \ .$$
	It is clear that $A_0,A_1,A_2$ form a partition of $\F_q^3$. Moreover, $|A_0|=q, |A_1|=3q(q-1)$ and $|A_2|=q(q-1)(q-2)$.
	Let $\tau\in \cT_{\rho,1,3}$ be any distribution. Following our standard reasoning, we have
	$\sum_{\vec x\in A_1}\tau(\vec x)+\sum_{\vec x\in A_2}2\tau(\vec x)\leq 3\rho$ and
	$\sum_{\vec x\in \F_q^3}\tau(\vec x)=1$.
	Under this condition, we try to upper bound
	$$
	H_q(\tau)=-\sum_{i=0}^2\sum_{\vec x\in A_i}\tau(\vec x)\log_q \tau(\vec x).
	$$
	Let $x_i=\sum_{\vec x\in A_i}\tau(\vec x)$ and the constraint becomes $x_0+x_1+x_2=1$, $x_1+2x_2\leq 3\rho$ and $x_0,x_1,x_2 \geq 0$.
	Then,
	$$
	-\sum_{\vec x\in A_i}\tau(\vec x)\log_q(\tau(\vec x))\leq -\left(\sum_{\vec x\in A_i}\tau(\vec x)\right)\log_q \left(\frac{\sum_{\vec x\in A_i}\tau(\vec x)}{|A_i|}\right)=-x_i\log_q\left(\frac{x_i}{|A_i|}\right)
	$$
	due to the concavity of the function $f(x)=x\log_q x$. Therefore,
	$$
	H_q(\tau)\leq 1+H_q(x_1,x_2)+x_1\log_q 3(q-1)+x_2\log_q (q-1)(q-2)
	$$
	subject to $x_1+2x_2\leq 3\rho$, $x_1+x_2 \leq 1$ and $x_1,x_2 \geq 0$. 
	
	We proceed to the lower bound argument. It suffices to find a distribution $\tau^*$ to reach $\max_{x_1+2x_2\leq 3\rho} \frac{1}{3}(1+H_q(x_1,x_2)+x_1\log_q 3(q-1)+x_2\log_q (q-1)(q-2))$.
	Let $x_1,x_2$ be the values to reach this maximum.
	Define the distribution $\tau$ such that
	$\tau^*(\vec x)=\frac{1-x_1-x_2}{q}$ for $\vec x\in A_0$, $\tau^*(\vec x)=\frac{x_1}{3q(q-1)}$ for $\vec x\in A_1$ and $\tau(\vec x)=\frac{x_2}{q(q-1)(q-2)}$ for $\vec x\in A_2$. Then, we have $\tau^*(A_0)=1-x_1-x_2, \tau^*(A_1)=x_1,\tau^*(A_2)=x_2$.
	We proceed to calculate $H_q(\tau^*)$.
	\begin{align*}
		H_q(\tau^*)&=(1-x_1-x_2)\log_q\frac{q}{1-x_1-x_2}+x_1\log_q \frac{3q(q-1)}{x_1}+x_2\log_q \frac{q(q-1)(q-2)}{x_2} \\
		&= 1+H_q(x_1,x_2)+x_1\log_q 3(q-1)+x_2\log_q (q-1)(q-2)
	\end{align*}
	Therefore, we conclude that
	\begin{align*}
		&\max_{\tau\in \cT_{\rho,1,3}}\frac{H_q(\tau)}{3} \geq \frac{H_q(\tau')}{3} \\
		&= \max\curl{\frac{1+H_q(x_1,x_2)+x_1\log_q 3(q-1)+x_2 \log_q (q-1)(q-2)}{3}:x_1+2x_2 \leq 3\rho, x_1+x_2 \leq 1, x_1,x_2 \geq 0}.
	\end{align*}
	The proof is completed.
\end{proof}

Again, by noting $\frac{1+F}{3} \geq \frac{F}{2}$ for all $F \leq 2$, we conclude that random linear codes do indeed perform better: the lower bound on the threshold rate furnished by Theorem~\ref{thm:thresholdrandomlinear} is strictly greater than the exact threshold rate of Theorem~\ref{thm:thresholdrandom}.

\subsection{List Decoding for Binary Alphabets with Larger Lists}

In this subsection, we observe that the list-decodability of random linear codes is better than random codes over the binary field for any list size $L$.

We begin by stating our possibility result for random linear codes. The proof is an adaptation of the argument from \cite{GuruswamiHSZ02,LiW18}.

\begin{theorem}\label{thm:randomlinear}
	For any fixed list size $L$ and $\delta>0$, a random linear code over the binary field of rate $1-h_2(\rho)-\frac{h_2(\rho)}{L-1-2\delta}-\delta$ is $(\rho, L)$-list decodable with probability $1-2^{-\Omega_{\delta,L}(n)}$.
\end{theorem}

For space reasons we just show that a random linear has positive probability of achieving the stated list-decodability, as is done in~\cite{GuruswamiHSZ02}; for the ``with high probability'' result the ideas used by~\cite{LiWootters} apply.

\begin{proof}
	Given a linear code $\cC\leq \F_2^n$, define the function
	$$
	S_\cC=2^{-n}\sum_{\vec x\in \F_2^n}2^{\frac{n}{L'}L_\cC(\vec x)}
	$$
	with $L'=\frac{L-1-2\epsilon}{H_2(\rho)}$ and $L_\cC(\vec x)=|B(\vec x,\rho)\cap \cC|$.
	It is clear that $S_{\{\vec 0\}}\leq 1+2^{n(H_2(\rho)+\frac{1}{L'}-1)}$.
	We define $\cC_0 = \{\vec 0\}$ and for $i \geq 1$, $\cC_i=\mathrm{span}_{\F_2}\{\vec v_1,\ldots,\vec v_i\}$, i.e, $\rlc_i$ is a random linear code that spanned by $\vec v_1,\ldots,\vec v_i \in \F_2^n$. 
	Given $\cC_{i-1}$, we now compute the expected value of $S_{\rlc_{i}}$, where $\rlc_i = \{0,\vec \bv_i\}+\cC_{i-1}$ for $\vec\bv_i \in \F_2^n$ sampled uniformly at random. 
	\begin{eqnarray*}
		&&\E[S_{\rlc_{i}}|\cC_{i-1}] = \E\left[\sum_{\vec x\in \F_2^n}2^{\frac{n}{L'}L_{\rlc_i}(\vec x)}|\cC_{i-1}\right] = 2^{-n}\sum_{\vec v_i\in \F_2^n}2^{-n}\sum_{\vec x\in \F_2^n}2^{\frac{n}{L'}L_{\cC_{i-1}}(\vec x)}\\
		&&=2^{-n}\sum_{\vec v_i\in \F_2^n}2^{-n}\sum_{\vec x\in \F_2^n}2^{\frac{n}{L'}L_{\cC_{i-1}}(\vec x)}\times 2^{\frac{n}{L'}L_{\cC_{i-1}}(\vec x+\vec v_i)}\leq S^2_{\cC_{i-1}}.
	\end{eqnarray*}
	Therefore, there exists $\vec v_i \in \F_2^n$ such that $S_{\cC_i}\leq S^2_{\cC_{i-1}}$. We continue in this manner to reach $\cC_k$ with
	$k=(1-h_2(\rho)-\frac{1}{L'}-\delta)n$. Then, we have
	$$
	S_{\cC_k}\leq S_{\cC_0}^{2^k}\leq (1+2^{n(h_2(\rho)+\frac{1}{L'}-1)})^{2^k}\leq \exp(2^{k-n+nh_2(\rho)+\frac{n}{L'}})=O(1) \ .
	$$
	On the other hand, we have that $\cC_k$ is $(\rho,L_{\max})$-list-decodable, where $L_{\max} = \max_{\vec x\in \F_2^n} L_{\cC_k}(\vec x)$. We now bound $L_{\max}$. Since $L_{\cC_k}(\vec x)=L_{\cC_k}(\vec x+\vec c)$ for any $\vec c\in \cC_k$, we have
	$$
	S_{\cC_k}=2^{-n}\sum_{\vec x\in \F_2^n}2^{\frac{n}{L'}L_{\cC_k}(\vec x)}\geq |\cC|2^{\frac{nL_{\max}}{L'}-n}=2^{k+\frac{nL_{\max}}{L'}-n}.
	$$
	Thus, we conclude that  $L_{\max}\leq \lfloor L'h_2(\rho)+1+\delta \rfloor=\lfloor L-\delta \rfloor=L-1$. This completes the proof.
\end{proof}

Next, we provide an upper bound on the list size of a random code. The proof uses the threshold framework.

\begin{theorem}\label{thm:randomcode}
	Let $L$ be a fixed constant list size and $\delta$ be any positive constant.
	With high probability, a random code with rate $\frac{L-1}{L}(1-h_2(\rho))-\frac{h_2(2\rho-2\rho^2)-h_2(\rho)}{L}+\delta$ is not $(\rho,L)$-list decodable.
\end{theorem}

\begin{proof}
	It suffices to bound the entropy $H_2(\tau)$ with $\tau=(\bx_1+\bz,\ldots,\bx_{L}+\bz)$ where
	$\bx_1,\ldots,\bx_{L}$ are independent random variables drawn according to $\Ber_2(0,\rho)$ and $\bz$ is a random variable drawn according to $\Ber_2(0,\frac{1}{2})$.
	\begin{eqnarray*}
		&&H_2(\bx_1+\bz,\bx_2+\bz,\ldots,\bx_{L+1}+\bz)=H_2(\bx_1+\bz,\bx_2+\bz)+H_2(\bx_3+\bz,\ldots,\bx_{L}+\bz|\bx_1+\bz,\bx_2+\bz)
		\\
		&&\geq H_2(\bx_1+\bz,\bx_2+\bz)+H_2(\bx_3+\bz,\ldots,\bx_{L}+\bz|\bx_1,\bz,\bx_2)\\
		&&=H_2(\bx_1+\bz,\bx_2+\bz)+\sum_{i=3}^{L}H_2(\bx_i+\bz|\bz)=H_2(\bx_1+\bz,\bx_2+\bz)+(L-2)h_2(\rho).
	\end{eqnarray*}
	It remains to bound $H_2(\bx_1+\bz,\bx_2+\bz)$. We notice that $\Pr[(\bx_1+\bz,\bx_2+\bz)=(0,0)]=\Pr[(\bx_1+\bz,\bx_2+\bz)=(1,1)]=\frac{1-2\rho+2\rho^2}{2}$ and $\Pr[(\bx_1+\bz,\bx_2+\bz)=(0,1)]=\Pr[(\bx_1+\bz,\bx_2+\bz)=(1,0)]=\rho-\rho^2$.
	This implies
	$$
		H_2(\bx_1+\bz,\bx_2+\bz)=1+h_2(2\rho-2\rho^2).
	$$
	This completes the proof.
\end{proof}

From these two theorems, we note the following. If we let $\delta$ tend to $0$, the upper bound provided by Theorem~\ref{thm:randomcode} is smaller than that provided by Theorem~\ref{thm:randomlinear} as $(3+\frac{1}{L-1})h_2(\rho)-h_2(2\rho-2\rho^2)<1$, assuming $\rho$ is not too large.

\bibliographystyle{alpha}
\bibliography{refs}

\appendix

\section{Proof of Claim~\ref{claim:sum-increases-entropy}} \label{sec:proof-of-sum-increases-entropy-claim}
In this section, we provide the proof of Claim~\ref{claim:sum-increases-entropy}, which we repeat here for convenience. 

\sumIncreasesEntropy*

In the following, for a finite set $S$, $\Unif(S)$ denotes the uniform distribution over $S$.

\begin{proof}
	We have 
	\[
		H_q(\bu_i+\beta\cdot\ba_j|\bS) = H_q(\bu_i+\beta\cdot\ba_j|\bS,\beta \cdot \ba_j) + I_q(\bu_i+\beta\cdot\ba_j;\beta\cdot\ba_j|\bS) \ .
	\]
	Note that
	\[
		H_q(\bu_i+\beta\cdot\ba_j|\bS,\beta \cdot \ba_j) = H_q(\bu_i|\bS) = h_{q,\ell}(p) \ .
	\]
	Thus, to conclude the theorem we need to prove that $I_q(\bu_i+\beta\cdot\ba_j;\beta\cdot\ba_j|\bS)>0$. By properties of mutual information, we have that $I_q(\bu_i+\beta\cdot\ba_j;\beta\cdot\ba_j|\bS)=0$ if and only if the random variables $(\bu_i+\beta\cdot\ba_j)|(\bS=S)$ and $\beta\cdot\ba_j|(\bS=S)$ are independent for all choices of $S \in \supp(\bS) = \binom{\F_q}{\ell}$. 
	
	To prove these random variables are \emph{not} independent, we show that the random variables $(\bu_i+\beta \cdot \ba_j|\bS=S,\beta \cdot \ba_j=\alpha)$ and $(\bu_i+\ba_j|\bS=S)$ do not follow the same distribution for any $\alpha \in \F_q^*$ for which $S+\alpha \neq S$. Note that as $\ba_j \sim \Unif(\F_q)$ and $\beta \neq 0$, we do have $\F_q^* \subseteq \supp(\beta\cdot\ba_j)$, and furthermore $S+\alpha = S$ for all $\alpha \in \F_q^*$ if and only if $S=\F_q$, but as $|S|=\ell < q$ we have $S \neq \F_q$. Now, note that 
	\[
		(\bu_i+\beta \cdot \ba_j|\bS=S,\beta \cdot \ba_j=\alpha) \sim p\cdot  \Unif(S+\alpha) + (1-p) \cdot \Unif(S+\alpha)
	\]
	and 
	\[
		(\bu_i+\beta\cdot\ba_j|\bS=S) \sim p\cdot\Unif(S) + (1-p)\cdot\Unif(S) \ .
	\]
	As $p < 1-\ell/q$ and $S+\alpha \neq S$, we conclude these distributions are distinct, as desired. 
\end{proof}

\end{document}